\documentclass[11pt]{article}
\usepackage[margin = 1in]{geometry}
\usepackage[utf8]{inputenc}

\usepackage{enumitem}
\setlistdepth{9}

\setlist[itemize,1]{label=$\bullet$}
\setlist[itemize,2]{label=$-$}
\setlist[itemize,3]{label=$*$}
\setlist[itemize,4]{label=$\cdot$}
\setlist[itemize,5]{label=$\bullet$}
\setlist[itemize,6]{label=$-$}
\setlist[itemize,7]{label=$*$}
\setlist[itemize,8]{label=$\cdot$}
\setlist[itemize,9]{label=$\bullet$}

\renewlist{itemize}{itemize}{9}

\newcommand{\Zeta}{Z}

\usepackage{amsmath,amssymb,amsthm,hyperref,color}
\usepackage{tikz}
\hypersetup{colorlinks=true,citecolor=blue, linkcolor=blue,
urlcolor=blue}
\usepackage{mdframed}

\newtheorem{theorem}{Theorem}
\newtheorem{lemma}[theorem]{Lemma}

\newtheorem{corollary}[theorem]{Corollary}
\newtheorem{remark}[theorem]{Remark}

\newtheorem{fact}[theorem]{Fact}
\newtheorem{definition}[theorem]{Definition}

\newtheorem*{lembranch}{Lemma~\ref{lem:branch}}
\newtheorem*{lempercolation}{Lemma~\ref{lem:percolation}}

\newtheorem*{lemtime}{Lemma~\ref{lem:time}}

\newcommand{\Eb}{\mathbf{E}}

\newcommand{\Ic}{\mathcal{I}}
\newcommand{\Psib}{\boldsymbol{\Psi}}
\newcommand{\Tmix}{T_\mathrm{mix}}
\newcommand{\Ent}{\mathrm{Ent}}
\newcommand{\emm}{\mathrm{e}}
\newcommand{\HardCore}{\textsc{Sample}}
\newcommand{\norm}[1]{\left\lVert#1\right\rVert}

\let\epsilon=\varepsilon

\title{Fast sampling via spectral independence\\beyond bounded-degree graphs\thanks{A preliminary version of the manuscript (without the proofs) appeared in the proceedings of \emph{ICALP 2022} (Track A). For the purpose of Open Access, the author has applied a CC BY public copyright licence to any Author Accepted Manuscript version arising from this submission. All data is provided in full in the results section of this paper.}} 
\author{
Ivona Bez\'{a}kov\'{a}\thanks{
  Department of Computer Science, Rochester Institute of Technology, Rochester, NY, USA. Supported in part by NSF grant DUE-1819546.} \and
Andreas Galanis\thanks{
  Department of Computer Science, University of Oxford, Wolfson Building, Parks Road, Oxford, OX1~3QD, UK.}
  \and
  Leslie Ann Goldberg$^\dag$
\and
 Daniel \v{S}tefankovi\v{c}\thanks{
Department of Computer Science, University of Rochester,
Rochester, NY 14627.  Research
supported by NSF grant CCF-2007287.}
}

\date{12 October 2023}
\begin{document}
\maketitle
\thispagestyle{empty}
\bibliographystyle{alpha}

\begin{abstract}
Spectral independence is a recently-developed framework for obtaining sharp bounds on the convergence time of the classical Glauber dynamics. This new framework has yielded optimal $O(n \log n)$  
sampling algorithms on bounded-degree graphs for a large class of problems throughout the so-called uniqueness regime,  including, for example, the problems of sampling independent sets, matchings, and Ising-model configurations.

Our main contribution is to relax the bounded-degree assumption that has so far been important in establishing and applying spectral independence. Previous methods for avoiding degree bounds rely on using $L^p$-norms to analyse contraction on graphs with bounded connective constant (Sinclair, Srivastava, Yin; FOCS'13). The non-linearity of $L^p$-norms is an obstacle to applying these results to bound spectral independence. Our solution is to capture the $L^p$-analysis recursively by amortising over the subtrees of the recurrence used to analyse contraction.  Our method generalises previous analyses that applied only to bounded-degree graphs.

As a main application of our techniques, we consider the random graph $G(n,d/n)$, 
where the previously known  algorithms run in time $n^{O(\log d)}$ or applied only to large $d$. We refine these algorithmic bounds significantly, and develop  fast nearly linear algorithms based on Glauber dynamics that apply to all constant $d$, throughout the uniqueness regime. 
\end{abstract}


\section{Introduction}
Spectral independence method was introduced by Anari, Liu, and Oveis Gharan \cite{Spectral} as a framework to obtain polynomial   bounds on the mixing time of Glauber dynamics. Originally based on a series of works on high-dimensional expansion \cite{high, KO20, O18, DK17, KM17}, it has since then been developed further using entropy decay by Chen, Liu, and Vigoda \cite{Optimal} who  obtained optimal $O(n \log n)$ mixing results on graphs of bounded maximum degree $\Delta$ whenever the framework applies. This paper focuses on relaxing the bounded-degree assumption of \cite{Optimal}, in sparse graphs where the maximum degree is not the right parameter  to capture the density of the graph. 


As a running example we will use the problem of sampling (weighted) independent sets, also known as the sampling problem from the hard-core model. For a graph $G=(V,E)$, the hard-core model with parameter $\lambda>0$ specifies a distribution $\mu_{G,\lambda}$ on the collection of independent sets of $G$, where for an independent set $I$  it holds that $\mu_{G,\lambda}(I)=\lambda^{|I|}/Z_{G,\lambda}$ where   $Z_{G,\lambda}$ is the partition function of the model (the normalising factor that makes the probabilities add up to~$1$).  For bounded-degree graphs of maximum degree $d+1$ (where $d\geq 2$ is an integer), it is known that the problems of   sampling and approximately counting 
 from this model
 undergo a computational transition at $\lambda_c(d)=\tfrac{d^d}{(d-1)^{d+1}}$, the so-called uniqueness threshold \cite{Weitz,Sly10,Inapprox2}: they are poly-time solvable when $\lambda<\lambda_c(d)$, and computationally intractable for $\lambda>\lambda_c(d)$. Despite this clear complexity picture, prior to the introduction of spectral independence, the algorithms for $\lambda<\lambda_c(d)$ were based on elaborate enumeration techniques whose running times scale as $n^{O(\log d)}$~\cite{Weitz,LLY13,PR,Peters}. The analysis of Glauber dynamics\footnote{Recall, for a graph $G$, the Glauber dynamics for the hard-core model iteratively maintains a random independent set $(I_t)_{t\geq 0}$, where at each step $t$ a  vertex $v$ is chosed u.a.r. and, if $I_t\cup \{v\}$ is independent,  it sets  $I_{t+1}=I_t\cup \{v\}$ with probability $\tfrac{\lambda}{\lambda+1}$, otherwise  $I_{t+1}=I_t\backslash \{v\}$. The mixing time is the maximum number (over the starting $I_0$) of steps $t$ needed to get within total variation distance 1/4 of $\mu_{G,\lambda}$, see Section~\ref{4f4f444f} for the precise definitions.} using spectral independence in the regime $\lambda<\lambda_c(d)$ yielded initially $n^{O(1)}$ algorithms for any $d$ \cite{Spectral} (see also \cite{touniqueness}), and then $O(n\log n)$ for bounded-degree graphs \cite{Optimal}. More recently, Chen, Feng, Yin, and Zhang \cite{RGlauber} obtained $O(n^2\log n)$ results  for arbitrary graphs $G=(V,E)$ that apply when $\lambda<\lambda_c(\Delta_G-1)$, where $\Delta_G$ is the maximum degree of $G$ (see also \cite{Spec2} for related results when $\Delta_G$ grows like $\log n$); this has been further refined to $O(n\log n)$ in \cite{AJKPV,chen2022localization, CFYZ22}.
 
The maximum degree is frequently a bad measure of the density of the graph, especially for graphs with unbounded-degree. One of the most canonical examples is the random graph $G(n,d/n)$ where the maximum degree grows with $n$ but the average degree is $d$, and therefore one would hope to be able to sample from $\mu_{G,\lambda}$ for $\lambda$ up to some constant, instead of $\lambda=o(1)$ that the previous results yield. In this direction, \cite{SSY,ConnConnPaper} obtained an $n^{O(\log d)}$ algorithm based on correlation decay that applies to all $\lambda<\lambda_c(d)$ for all graphs with ``connective constant'' bounded by $d$ (meaning, roughly, that for all $\ell=\Omega(\log n)$ the number of length-$\ell$ paths starting from any vertex is bounded by $d^{\ell}$). 
 The result of~\cite{ConnConnPaper} applies to $G(n,d/n)$ for all $d>0$. In terms of Glauber dynamics on $G(n,d/n)$, \cite{MosselSly} showed an  $n^{1+\Omega(1/\log \log n)}$ lower bound on the mixing time in the case of the Ising model; this lower bound actually applies to most well-known models, and in particular rules out $O(n \log n)$ mixing time results for the hard-core model when $\lambda=\Omega(1)$. The mixing-time lower bound on $G(n,d/n)$ has only been matched by complementary fast mixing results in models with strong monotonicity properties, see \cite{MosselSly} for the ferromagnetic  Ising model and \cite{Potts} for the random-cluster model. Such monotonicity properties unfortunately do not hold for the hard-core model, and the best known results \cite{GL1,GL2} for Glauber dynamics on $G(n,d/n)$ give an $n^{C}$ algorithm for $\lambda<1/d$ and sufficiently large $d$ (where $C$ is a constant depending on $d$).
 Our goal in this paper is to go all the way up to $\lambda_c(d)$ (which converges to $e/d$ so is larger than $1/d$).
 
 Our main contribution is to obtain nearly linear-time algorithms on $G(n,d/n)$, for 
 all of the models considered in~\cite{ConnConnPaper}, i.e., the hard-core model, the monomer-dimer model (weighted matchings), and the antiferromagnetic Ising model. Key to our results are new spectral independence bounds for any $d>0$ in the regime $\lambda<\lambda_c(d)$ for arbitrary graphs $G=(V,E)$ in terms of their ``$d$-branching value'' (which resembles the connective-constant notion of \cite{ConnConnPaper}). To state our main theorem for the hard-core model on $G(n,d/n)$, we first extend the definition of $\lambda_c(d)$ to all reals $d>0$ by setting $\lambda_c(d)=\tfrac{d^{d}}{(d-1)^{d+1}}$ for $d>1$, and $\lambda_c(d)=\infty$ for $d\in (0,1)$.  We use the term ``\emph{whp} over  the choice of $G\sim G(n,d/n)$'' as a shorthand for ``as $n$ grows large, with probability $1-o(1)$ over the choice of $G(n,d/n)$''. An $\epsilon$-sample from a distribution $\mu$ supported on a finite set $\Omega$ is a random  $\sigma\in \Omega$ whose distribution $\nu$ satisfies $\left\|\nu-\mu\right\|_{\mathrm{TV}}\leq \epsilon$, where $\left\|\nu-\sigma\right\|_{\mathrm{TV}}=\tfrac{1}{2}\sum_{\sigma\in\Omega} |\nu(\sigma)-\mu(\sigma)|$.

\begin{theorem}\label{thm:main}
Let $d,\lambda>0$ be such that $\lambda<\lambda_c(d)$. For any arbitrarily small constant $\theta>0$, there is an algorithm such that, whp over the choice of $G\sim G(n,d/n)$, 
when the algorithm is given as input
 the graph $G$ and an arbitrary rational $\epsilon>0$, it  outputs an  $\epsilon$-sample from $\mu_{G,\lambda}$ in time $n^{1+\theta}\log \tfrac{1}{\epsilon}$.
\end{theorem}

The reader might wonder why is there no constant in front of the running time (in Theorems~\ref{thm:main}, \ref{thm:main1}, and \ref{thm:main2}) or why is there no requirement that $n$ is sufficiently large? The assumption 
that $n$ is sufficiently large is taken care of in the whp condition: there is a function $f_{d,\lambda,\theta}:{\mathbb{Z}}\rightarrow {\mathbb{R}}$ such that $\lim_{n\rightarrow\infty} f_{d,\lambda,\theta}(n) = 0$ and the ``whp'' means with probability $\geq 1 - f_{d,\lambda,\theta}(n)$; the function $f_{d,\lambda,\theta}$ equals 1 for small $n$ (making the conclusion trivial for such $n$). Moreover, the family of $O(n^{1+\theta})$ algorithms from Theorem~\ref{thm:main} can be turned into an $n^{1+o(1)}$ algorithm, see Remark~\ref{remark} for a discussion.

We remark also here that the algorithm of Theorem~\ref{thm:main} (as well as Theorems~\ref{thm:main1} and~\ref{thm:main2} below) can also recognise in time $n^{1+o(1)}$ whether the graph $G\sim G(n,d/n)$ is a ``good'' graph, i.e., we can formulate graph properties that guarantee the success of the algorithm,  are satisfied whp, and are  also efficiently verifiable, see Section~\ref{sec:verify} for details.

The key to obtaining Theorem~\ref{thm:main} is to bound the spectral independence of the Gibbs distribution on~$G(n,d/n)$. The main strategy that has been applied so far to bound spectral independence is to adapt suitably correlation decay arguments and, therefore, it is tempting to use the correlation decay analysis of \cite{ConnConnPaper}. This poses new challenges in our setting since \cite{ConnConnPaper} uses an $L^p$-norm analysis of correlation decay on trees, and the non-linearity of $L^p$-norms is an obstacle to converting their analysis into spectral independence bounds (in contrast, for bounded-degree graphs, the $L^\infty$-norm is used which can be converted to spectral independence bounds using a purely analytic approach, see \cite{touniqueness}). Our solution to work around that is to ``linearise'' the $L^p$-analysis by taking into account the structural properties of subtrees. This allows us to amortise over the tree-recurrence using appropriate combinatorial information (the $d$-branching values) and to  bound subsequently spectral independence; details are given in Section~\ref{sec:spectral}, see Lemmas~\ref{lem:spectraltrees} and~\ref{lem:hardspectral} (and equation~\eqref{eq:spectralind} that is at the heart of the argument). Once the spectral independence bound is in place, further care is needed to obtain the fast nearly linear running time, paying special attention to the distribution of high-degree vertices inside $G(n,d/n)$ and to blend this with the entropy-decay tools developed in \cite{Optimal}, see Section~\ref{sec:entropyfacb} for this part of the argument.
 
 %
 %

In addition to our result for the
hard-core model, we also obtain similar results for  the Ising 
model and te Momomer-Dimer model.
The configurations of the Ising model on 
a graph $G=(V,E)$ are assignments 
$\sigma\in \{0,1\}^V$ which assign the spins~$0$
and~$1$ to the vertices of~$G$.
The   Ising model with parameter $\beta>0$ corresponds to a distribution $\mu_{G,\beta}$ on $\{0,1\}^V$, where for an assignment $\sigma\in \{0,1\}^V$, it holds that $\mu_{G,\beta}(\sigma)=\beta^{m(\sigma)}/Z_{G,\beta}$ where $m(\sigma)$ is the number of edges whose endpoints have the same spin assignment under $\sigma$, and   $Z_{G,\beta}$ is the partition function of the model. The model is antiferromagnetic when $\beta\in (0,1)$, and ferromagnetic otherwise. For $d\geq 1$, let $\beta_c(d)=\tfrac{d-1}{d+1}$; for $d\in (0,1)$, let $\beta_c(d)=0$. It is known  that on bounded-degree graphs of maximum degree $d+1$ the sampling/counting problem for the antiferromagnetic Ising model undergoes a phase transition at $\beta=\beta_c(d)$, analogous to that for the hard-core model \cite{SST,LLY13, SlySun,Inapproxising}.
\begin{theorem}\label{thm:main1}
Let $d,\beta>0$ be  such that $\beta\in (\beta_c(d),1)$. For any constant $\theta>0$, there is an algorithm such that, whp over the choice of $G\sim G(n,d/n)$,
when the algorithm is given as input the graph $G$ and an arbitrary rational $\epsilon>0$, it outputs an  $\epsilon$-sample from $\mu_{G,\beta}$ in time $n^{1+\theta}\log \tfrac{1}{\epsilon}$.
\end{theorem}

For a graph $G=(V,E)$, the monomer-dimer model with parameter $\gamma>0$ corresponds to a distribution $\mu_{G,\gamma}$ on the set of matchings of $G$, where for a matching $M$, it holds that $\mu_{G,\gamma}(M)=\gamma^{|M|}/Z_{G,\gamma}$ where $Z_{G,\gamma}$ is the partition function. For general graphs $G=(V,E)$ and $\gamma=O(1)$, \cite{jer,jerb} gave an $O(n^2m\log n)$ algorithm (where $n=|V|, m=|E|$), which was improved for bounded-degree graphs to $O(n\log n)$ in \cite{Optimal} using spectral independence. For $G(n,d/n)$, \cite{ConnConnPaper} gave an $O(n^{\log d})$ deterministic algorithm using correlation decay, and \cite{Spec2} showed that Glauber dynamics mixes in $n^{2+o(1)}$ steps when $\gamma=1$.
\begin{theorem}\label{thm:main2}
Let $d,\gamma>0$. For any constant $\theta>0$, there is an algorithm such that, whp over the choice of $G\sim G(n,d/n)$,
when the algorithm is given as input the graph $G$ and an arbitrary rational $\epsilon>0$, it  outputs an  $\epsilon$-sample from $\mu_{G,\gamma}$ in time $n^{1+\theta}\log \tfrac{1}{\epsilon}$.
\end{theorem}

In the next section, we give the main ingredients of our algorithm for the hard-core model and we give the proof of Theorem~\ref{thm:main}. The proofs of Theorems~\ref{thm:main1} and~\ref{thm:main2} build on similar ideas, though there are some modifications needed to obtain the required spectral independence bounds. We give their proofs in Section~\ref{sec:remainingproofs}. 

Before proceeding let us finally mention that, to go beyond the 2-spin models studied here, the main obstacle is to establish the spectral independence bounds for graphs with potentially unbounded degrees. As it is pointed out in \cite[Section 7]{ConnConnPaper}, their correlation-decay analysis does not extend to other models in a straightforward manner, and hence it is natural to expect that the same is true for spectral independence as well.

\subsection{Further developments}

Our algorithms are based on running Glauber dynamics on (relatively) low-degree vertices. Subsequent to our work, Efthymiou and Feng \cite{EftFeng} obtained for the hard-core model an $n^{1+O(1/\log \log n)}$ mixing-time bound for Glauber dynamics on $G(n,d/n)$ when $\lambda<\lambda_c(d)$ (and similarly for the monomer-dimer model), i.e., without the need to restrict to low-degree vertices. Their spectral independence arguments build upon the ``linearisation'' of the $L^p$ analysis we introduce here, which are then combined with the framework of \cite{RGlauber} to obtain the improved mixing-time bounds.

\section{Proof outline for Theorem~\ref{thm:main}}\label{sec:outline}

Our algorithm for sampling from the hard-core model on a graph $G=(V,E)$  is an adaptation of Glauber dynamics on an appropriate set of ``small-degree'' vertices $U$, the details of the algorithm are given in Figure~\ref{alg:gl}. Henceforth, analogously to the Ising model, it will be convenient to view the hard-core model as a 2-spin model supported on $\Omega\subseteq \{0,1\}^V$, where $\Omega$ corresponds to the set of independent sets of $G$ (for an independent set $I$, we obtain $\sigma\in \{0,1\}^V$ by setting $\sigma_v=1$ iff $v\in I$).
\begin{figure}[h]
\begin{mdframed}
\textbf{Algorithm} \HardCore$(G,T)$ \vskip 0.05cm
\rule[0.3cm]{6cm}{0.4pt}
\vskip 0.05cm

\noindent \textbf{Parameters:} $D>0$ (threshold for small/high degree vertices). \vskip 0.25cm

\noindent \textbf{Input:} \phantom{ \ \ } Graph $G=(V,E)$, integer $T\geq 1$ (number of iterations).\vskip 0.15cm

\begin{description}
    \item[\textbf{1. Initialisation:}] Let $U$ be the set of all vertices with degree $\leq D$.
    
    Let $X_0$ be the empty independent set on $U$.
    \item[2. Main loop:] For $t=1,\hdots, T$,
        \begin{itemize}
    \item Pick a vertex $u$ uniformly at random from $U$.
    \item For every vertex $v\in U\backslash\{u\}$, set $X_t(v)=X_{t-1}(v)$.
    \item Sample the spin $X_t(u)$ according to $\mu_{G,\lambda}\big(\sigma_u\mid \sigma_{U\backslash\{u\}}=X_t(U\backslash\{u\})\big)$, i.e., update $u$ according to the hard-core distribution on the \emph{whole graph} $G$, conditioned on the spins of $U\backslash\{u\}$.

    \end{itemize}
    \item[3. Finalisation:] Sample $\sigma\sim\mu_{G,\lambda}\big(\cdot\,\big|\, \sigma_U=X_T\big)$, i.e., extend $X_T$ to the whole vertex set of $G$ by sampling from $\mu_{G,\lambda}$ conditioned on the configuration on $U$.
\end{description}
\end{mdframed}
\caption{\label{alg:gl} The \HardCore$(G,T)$ subroutine for sampling from the hard-core distribution $\mu_{G,\lambda}$. We use the analogue of this algorithm for the Ising model with parameter $\beta$ (replacing $\mu_{G,\lambda}$ by $\mu_{G,\beta}$). For the monomer-dimer model, the only difference is that the algorithm needs to update (single) edges in $F$, where $F$ is the set of edges whose both endpoints lie in $U$ (i.e., degree $\leq D$).}
\end{figure}
Note that for general graphs $G$,   implementing Steps 2 and Steps 3 of the algorithm might be difficult. The following lemma exploits the sparse structure of $G(n,d/n)$ and in particular the fact that high-degree vertices are sparsely scattered. We will use this in the proof of our main theorems to show that the algorithm \HardCore$(G,T)$ can be implemented very efficiently for appropriate $D$, paying only $O(\log n)$ per loop operation in Step 2 and only $O(n\log n)$ in Step 3.  The \emph{tree-excess} of a graph $G=(V,E)$ is defined as $|E|-|V|+1$.
\newcommand{\statelemtime}{Let $d>0$ be an arbitrary real. There exist constants $D,\ell>0$ such that the following holds whp over the choice of $G=(V,E)\sim G(n,d/n)$. Each of the connected components of $G[V\backslash U]$, where $U$ is the set of vertices of degree $\leq D$, has size $O(\log n)$ and tree-excess at most~$\ell$.}
\begin{lemma}\label{lem:time}
\statelemtime
\end{lemma}

Lemma~\ref{lem:time} follows using relatively standard techniques from random graphs and is proved in Section~\ref{sec:randomproofs}. Later, we will establish a more refined version of this  property that will allow us to bound the mixing time of the single-site dynamics 
that we consider (the main loop of
 \HardCore$(G,T)$).

The key ingredient needed to prove our main result is to show that the main loop of our sampling algorithm returns a good sample on the induced hard-core distribution on the set $U$. More precisely, for a graph $G=(V,E)$ and $U\subseteq V$, we let $\mu_{G,\lambda,U}(\cdot)$ denote the induced distribution on the spins of $U$, i.e., the marginal distribution $\mu_{G,\lambda}(\sigma_U=\cdot)$.
\newcommand{\statelemfastone}{Let $d,\lambda>0$ be constants such that $\lambda<\lambda_c(d)$. For any arbitrarily small constant $\theta>0$, there is $D>0$ such that the following holds whp over the choice of $G\sim G(n,d/n)$. 

Let $U$ be the set of vertices in $G$ of degree $\leq D$.  Then, for any $\epsilon>0$, for $T=\lceil n^{1+\theta/2}\log \tfrac{1}{\epsilon}\rceil$, the main loop of \HardCore$(G,T)$ returns a sample $X_T$ from a distribution which is $\epsilon$-close to $\mu_{G,\lambda, U}$.}
\begin{lemma}\label{lem:fastone}
\statelemfastone
\end{lemma}
We will prove Lemma~\ref{lem:fastone} in Section~\ref{sec:entropyfacb}.
With these two lemmas we are ready to prove Theorem~\ref{thm:main}.
\begin{proof}[Proof of Theorem~\ref{thm:main}]

We give first the details for the more interesting case  $d\geq 1$. Consider arbitrarily small $\theta>0$ and $D,\ell$  as in Lemmas~\ref{lem:time} and~\ref{lem:fastone}, so that whp $G$ satisfies the properties therein.  
(It is clear from the proofs of Lemmas~\ref{lem:time} and~\ref{lem:fastone} that $D$ and $\ell$ can be computed given $d$, $\lambda$, and $\theta$.)
Let $\epsilon>0$ be the desired accuracy for sampling from $\mu_{G,\lambda}$; it is sufficient to consider $\epsilon<1/\emm$. Let  $U$ be the set of vertices with degree $\leq D$, and set $T=\lceil n^{1+\theta/2}\log\tfrac{1}{\epsilon}\rceil$. 

By Lemma~\ref{lem:fastone}, whp over the choice of $G$, the main loop of \HardCore$(G,T)$ returns a configuration $X_T:U\rightarrow \{0,1\}$ that is $\epsilon$-close to $\mu_{G,\lambda, U}$. Note that each iteration of the main loop of \HardCore$(G,T)$ can be implemented in $O(\log n)$ time since $G[V\backslash U]$ has components of size $O(\log n)$ and tree excess at most $\ell$. In particular, any vertex $u\in U$ can be adjacent to at most $D$ of these components, and therefore the component of $u$ in $G[(V\backslash U)\cup \{u\}]$ has size $O(\log n)$ and tree excess at most $k=D\,\lceil\ell\rceil=O(1)$. We can therefore sample the spin of $u$ under $\mu_{G,\lambda}$ conditioned on the spins of $U\backslash \{u\}$ in time $O(4^{k}\log n)=O(\log n)$.\footnote{One ``naive'' way to do this is by considering a spanning tree and then brute-forcing over all $\leq 4^k$ possibilities for the endpoints of the excess edges (the spins on each edge can be set in at most 4 ways). For each of these, the marginal probability at $u$ and the corresponding partition function can be computed using dynamic programming on the left-over tree.} Therefore, the main loop of \HardCore$(G,T)$ runs in time $O(T\log n)$. Analogously, the finalisation step of \HardCore$(G,T)$, i.e., extending the configuration $X_T$ on $U$ to a configuration $\sigma$ on the whole vertex set $V$, can be implemented in time $O(n\log n)$ by iterating over the vertices in $V\backslash U$ and using the fact that the components of $G[V\backslash U]$ have excess at most $\ell$. Therefore, the overall running time of the algorithm is bounded by $O(T\log n)+O(n \log n)$, which is less than $\lceil n^{1+\theta}\log\tfrac{1}{\epsilon}\rceil $ for all sufficiently large $n$. It remains to note that, since $X_T$ is $\epsilon$-close to the marginal distribution of $\mu_{G,\lambda}$ on $U$, and the finalisation step is done perfectly conditioned on the configuration on $U$, the final configuration $\sigma$ is $\epsilon$-close to the distribution $\mu_{G,\lambda}$.

For $d<1$, whp $G$ consists of tree-like components of size $O(\log n)$, and therefore we can obtain a perfect sample from $\mu_{G,\lambda}$ in time $O(n\log n)$ by
going through the vertices one by one and, 
for each vertex, taking $O(\log n)$ time to compute its marginal, conditioned
on the spins already sampled.
\end{proof}

\section{Spectral independence via branching values}\label{sec:spectral}

We first introduce the notions of spectral independence and pairwise vertex influences, which we will later use to bound the mixing time of the main loop of \HardCore$(G,T)$, i.e., to prove Lemma~\ref{lem:fastone}. We will define the terminology in a general way that will be useful both for our analysis of the hard-core model,
and for our later analysis of other models. 

Let $q\geq 2$ be an integer indicating the number of spins and let $V$ be a set of size $n$. We consider distributions $\mu$ supported on a set $\Omega\subseteq [q]^V$.\footnote{For an integer $k\geq 1$, we denote by $[k]$ the set $\{0,1,\hdots,k-1\}$.} For $S\subseteq V$, let $\Omega_S=\{\tau\in [q]^S \mid \mu(\sigma_S=\tau)>0\}$ be the set of all partial configurations on $[q]^S$ that have non-zero marginal under $\mu$. For $\tau\in \Omega_S$, let  $\mu_\tau$ be the conditional distribution on $\Omega$ induced by $\tau$, i.e., $\mu_\tau(\cdot)=\mu(\cdot\mid \sigma_S=\tau)$. Let $\mu_{\min}=\min_{\sigma\in \Omega}\mu(\sigma)$.

For $S\subseteq V$ and $\tau\in \Omega_S$, the \emph{influence matrix} conditioned on $\tau$ is the matrix $\Psib_\tau$ whose rows and columns are indexed by $\tilde{V}_\tau=\{(v,i)\mid  v\in V\backslash S,\,  \mu_\tau(\sigma_v=i)>0\}$, where the entry indexed by $(v,i), (w,k)$ equals $\mu_\tau(\sigma_w=k\mid \sigma_v=i)-\mu_\tau(\sigma_w=k)$ if $v\neq w$, and 0 otherwise. It is a standard fact that the eigenvalues of the matrix $\Psib$ are all real (\cite{Spectral}), and we denote by  $\lambda_1(\Psib)$ its largest eigenvalue.

\begin{definition}
Let $q\geq 2$ be an integer and $V$ be a set of size $n\geq 2$.  Let $\mu$ be a distribution supported over $\Omega\subseteq [q]^V$. Let $\eta,b>0$.
We say that $\mu$ is \emph{$\eta$-spectrally independent} if for all $S\subset V$ and $\tau\in \Omega_{S}$, it holds that $\lambda_1(\Psib_\tau)\leq \eta$.
We say that $\mu$ is \emph{$b$-marginally bounded} if for all $S\subset V$, $v\in V\backslash S$, $\tau\in \Omega_{S}$, and $i\in [q]$,  it either holds that  $\mu_\tau(\sigma_v=i)=0$ or else $\mu_\tau(\sigma_v=i)\geq b$.
\end{definition}

Following \cite{Spectral, touniqueness}, for distributions $\mu$ induced by 2-spin systems, we work with the following notion of pairwise vertex-influence, which can be used to bound the spectral independence. For a graph $G=(V,E)$ and $\tau\in \{0,1\}^S$ for some $S\subset V$, for vertices $u,v$ with $u\in V\backslash S$ and $0<\mu_\tau(\sigma_u=1)<1$, we define the \emph{influence} of $u$ on $v$ (under $\mu_\tau$) as
\[\Ic^{\tau}_G(u\rightarrow v)=\mu_\tau(\sigma_v=1\mid \sigma_u=1)-\mu_\tau(\sigma_v=1\mid \sigma_u=0).\]
For matchings, we will work with an analogous notion from the perspective of edges (see Section~\ref{sec:monomer}). For all these models, spectral independence will be bounded by summing the absolute value  of the influences of an arbitrary vertex $u$ to the rest of the graph.
In turn, it has been shown in \cite{touniqueness} that summing the influences of a vertex $u$ in a graph $G$ reduces to summing the sum of influences on the self-avoiding walk tree emanating from $u$ (this is the tree whose vertices correspond to self-avoiding walks in $G$ which start from $u$, where two vertices are adjacent when one walk is a one-step extension of the other; it was first introduced in related contexts by \cite{Weitz,scott2005repulsive}). 

We will use the following lemma from \cite{touniqueness}. 
\begin{lemma}[{\cite{touniqueness}}]\label{lem:graphtotrees}
Consider an arbitrary 2-spin system on a graph $G=(V,E)$, with distribution $\mu$ supported on $\Omega\subseteq \{0,1\}^V$. 
Let $\rho\in V$ be an arbitrary vertex, $T=(V_T,E_T)$ be the tree of self-avoiding walks in $G$ starting from $\rho$, and $\nu$ be the distribution of the 2-spin system on $T$.

Then, for any $S\subseteq V\backslash \{\rho\}$ and $\tau\in \Omega_S$ with $0<\mu_\tau(\sigma_\rho=1)<1$,    there is a subset $W\subseteq V_T\backslash \{\rho\}$ and a configuration $\phi\in \{0,1\}^{W}$  such that $\mu_\tau(\sigma_\rho=1)=\nu_\phi(\sigma_\rho=1)$ and
\[\sum_{v\in V}\big|\Ic^{\tau}_G(\rho\rightarrow v)\big|\leq \sum_{u\in V_T} \big|\Ic^{\phi}_{T}(\rho\rightarrow u)\big|,\]
where $\Ic^{\phi}_{T}(\rho\rightarrow \cdot)$ denotes the influence of $\rho$ on the vertices of $T$ under $\nu_\phi$.
\end{lemma}

\subsection{The branching value}
We will need the following notion to capture the growth of the self-avoiding walk tree from a vertex.
\begin{definition}\label{def:dbranchingvalue}
Let $d>0$ be a real number and $G=(V,E)$ be a graph. For a vertex $v$ in $G$, the $d$-branching value $S_v$ equals $\sum_{\ell\geq 0} N_{v,\ell}/d^{\ell}$, where $N_{v,\ell}$ is the number of (simple) paths with a total of $\ell+1$ vertices starting from $v$ (for convenience, we set $N_{v,0}=1$). 
\end{definition}
We will show the following lemma in Section~\ref{sec:branch} which bounds the $d'$-branching value of $G(n,d/n)$ for any $d'>d$.
\newcommand{\statelembranch}{Let $d\geq 1$. Then, for every $d'>d$ and $\epsilon>0$, whp over the choice of $G\sim G(n,d/n)$, the $d'$-branching value of every vertex in $G$ is at most $\epsilon \log n$.}
\begin{lemma}\label{lem:branch}
\statelembranch
\end{lemma}

\subsection{Spectral independence for the hard-core model}\label{sec:ffferfr}
 
In this section, we bound the spectral independence of  $G(n,d/n)$ in the hard-core model when $\lambda<\lambda_c(d)$.   We will need the following technical lemma that can be derived from \cite{ConnConnPaper}. The derivation details are similar to an analogous lemma for matchings (cf. Lemma~\ref{lem:connconst2} below), which can be found in \cite[Lemma 15]{matchings}.
\begin{lemma}[\cite{ConnConnPaper}]\label{lem:connconst}
Let $d>1$ and $\lambda>0$ be constants such that $\lambda<\lambda_c(d)$. Let $\chi\in (1,2)$ be given from $
\frac{1}{\chi} = 1 - \frac{d-1}{2}\log\left(1+\frac{1}{d-1}\right)
$ and set $a=\tfrac{\chi}{\chi-1}$. Consider also the function  $\Phi(y)=\frac{1}{\sqrt{y(1+y)}}$ for $y>0$. Then, there is a constant $0<\kappa<1/d$ such that the following holds for any integer $k\geq 1$.

Let $x_1,\hdots, x_k>0$ be reals and $x= \lambda\prod_{i=1}^k \frac{1}{1+x_i}$. Then $(\Phi(x))^a\sum^k_{i=1}\Big(\frac{x}{(1+x_i)\Phi(x_i)}\Big)^a\leq \kappa^{a/\chi}$.
\end{lemma}

We will show the following.
\begin{lemma}\label{lem:spectraltrees}
Let $d> 1$ and $\lambda>0$ be constants such that $\lambda<\lambda_c(d)$. Then, there is a constant  $\chi>1$ such that the following holds.

Let $T=(V,E)$ be a tree rooted at $\rho$, whose $d$-branching value is  $\leq \alpha$ and whose root has $k$ children. Then, for the hard-core distribution on $T$ with parameter $\lambda$, any $S\subseteq V\backslash \{\rho\}$ and $\tau\in \Omega_S$ with $0<\mu_\tau(\sigma_\rho=1)<1$, it holds that
\[\sum_{v\in V}\big|\Ic^{\tau}_T(\rho\rightarrow v)\big|\leq W_k \alpha^{1/\chi},\]
where $W_k>0$ is a real depending only on the degree $k$ of the root (and the constants $d,\lambda$).
\end{lemma}
\begin{proof}
Let $\kappa\in (0,1/d)$ and $\chi\in (1,2)$ be the constants from Lemma~\ref{lem:connconst}, and  $\Phi(x)=\frac{1}{\sqrt{x(1+x)}}$ be also as in Lemma~\ref{lem:connconst}. We may assume without loss of generality that $S$ is empty (and $\tau$ is trivial) by truncating the tree $T$ using the following procedure: just remove vertices $u\in S$ with $\tau_u=0$, and for $u\in S$ with $\tau_u=1$ remove $u$ and all of its neighbours. Note that for all the removed vertices $v$ it holds that $\Ic^{\tau}_T(\rho\rightarrow v)=0$, so the removal procedure does not decrease the sum of the absolute influences, while at the same time decreasing the $d$-branching value of the tree  $T$. Henceforth, we will drop $\tau$ and $S$ from notation. 

To prove the lemma, we will work inductively on the depth of the tree. To this end, we first define for each vertex $u$ in $T$ the following values $\alpha_u$ and $R_u$; the $\alpha$'s capture a rooted analogue of the branching value of internal vertices within $T$, while the $R$'s the marginals of the vertices in the corresponding subtrees. More precisely, if $v$ is a leaf, set $\alpha_v=1$ and $R_v=\lambda$; otherwise set $\alpha_v=1+\tfrac{1}{d}\sum^t_{i=1}\alpha_{v_i}$ and $R_v=\lambda\prod^t_{i=1}\tfrac{1}{1+R_{v_i}}$, where $v_1,\hdots, v_t$ are the children of $v$. Note that for the root $\rho$ we have that $\alpha_\rho=S_\rho\leq \alpha$, where $S_\rho$ is the $d$-branching value of $\rho$ in the tree $T$. Moreover, if we denote by $T_v$ the subtree of $T$ rooted at $v$ and by $u$ the parent of $v$ in $T$, then it holds that
\begin{equation}\label{eq:ratios}
R_v=\frac{\mu_{T_v,\lambda}(\sigma_v=1)}{\mu_{T_v,\lambda}(\sigma_v=0)} \quad\mbox{and}\quad \Ic_{T}(u\rightarrow v)=-\frac{R_v}{R_v+1}.
\end{equation}
The first equality is fairly standard and can be proved using induction on the height of the tree, while the second one is \cite[Lemma 15]{touniqueness} (it also follows directly from the definition of influence and the first equality).

 For an integer $h\geq 0$, let $L(h)$ be the nodes at distance $h$ from the root $\rho$. Let $M_k=\sqrt{1+(1+\lambda)^k/\lambda}$, where recall that $k$ is the degree of the root $\rho$. We will show that
\begin{equation}\label{eq:spectralind}
\sum_{v\in L(h)}\Big(\frac{\alpha_v}{\alpha_\rho}\Big)^{1/\chi}\frac{\big|\Ic_T(\rho\rightarrow v)\big|}{R_v\Phi(R_v)}\leq M_k (d\kappa)^{h/\chi}.
\end{equation}
Since $\alpha_v\geq 1$ for $v\in V$, $\alpha_\rho\leq \alpha$ and $R_v\Phi(R_v)\leq 1$, \eqref{eq:spectralind} yields $\sum_{v\in L(h)}\big|\Ic_T(\rho\rightarrow v)\big|\leq M_k\alpha^{1/\chi} (d\kappa)^{h/\chi}$ for all integer $h\geq 0$, and therefore summing over $h$, we obtain that
\[\sum_{v\in V}\big|\Ic_T(\rho\rightarrow v)\big|\leq M_k\alpha^{1/\chi}\sum_{h\geq 0}(d\kappa)^{h/\chi}\leq \frac{M_k\alpha^{1/\chi}}{1-(d\kappa)^{1/\chi}},\]
which proves the result with $W_k=\tfrac{M_k}{1-(d\kappa)^{1/\chi}}$. So it only remains to prove \eqref{eq:spectralind}.

We will work inductively.  The base case $h=0$ is equivalent to 
$M_k\geq 1/(R_\rho \Phi(R_\rho))=\sqrt{1+1/R_\rho}$, which is true since from the recursion \eqref{eq:ratios} for $R_\rho$ we have that $R_\rho\geq \lambda/(1+\lambda)^k$ (using the trivial bound $R_v\leq \lambda$ for each $v$). For the induction step, consider $v\in L(h-1)$ and suppose it has $k_v\geq 0$ children, denoted by $v_i$ for $i\in [k_v]$. Then, for each $i\in [k_v]$, since $v$ is on the unique path joining $\rho$ to $v_i$, it holds that (see \cite[Lemma B.2]{Spectral})
\[\Ic_T(\rho\rightarrow v_i)=\Ic_{T}(\rho\rightarrow v)\Ic_T(v\rightarrow v_i),\]
so we can write
\begin{equation}\label{eq:writeout}
\sum_{v\in L(h)}\Big(\frac{\alpha_v}{\alpha_\rho}\Big)^{1/\chi}\frac{\big|\Ic_T(\rho\rightarrow v)\big|}{R_v\Phi(R_v)}=\sum_{v\in L(h-1)}\Big(\frac{\alpha_v}{\alpha_\rho}\Big)^{1/\chi}\frac{|\Ic_T(\rho\rightarrow v)|}{R_v\Phi(R_v)}\sum_{i\in [k_v]}\Big(\frac{\alpha_{v_i}}{\alpha_v}\Big)^{1/\chi}R_v\Phi(R_v)\frac{|\Ic_T(v\rightarrow v_i)|}{R_{v_i}\Phi(R_{v_i})}.
\end{equation}

Consider an arbitrary $v\in L(h-1)$. Then, since $\tfrac{1}{\chi}+\tfrac{1}{a}=1$, by H\"{o}lder's inequality we have that
\begin{equation}\label{eq:induction}
\sum_{i\in [k_v]}\Big(\frac{\alpha_{v_i}}{\alpha_v}\Big)^{1/\chi}R_v\Phi(R_v)\frac{|\Ic_T(v\rightarrow v_i)|}{R_{v_i}\Phi(R_{v_i})}\leq \bigg(\sum_{i\in [k_v]}\frac{\alpha_{v_i}}{\alpha_v}\bigg)^{1/\chi}\bigg((R_v\Phi(R_v))^a\sum_{i\in [k_v]}\Big(\frac{|\Ic_T(v\rightarrow v_i)|}{R_{v_i}\Phi(R_{v_i})}\Big)^a\bigg)^{1/a}.
\end{equation}
Note that for $x=R_v$ and $x_i=R_{v_i}$, $i\in[k_v]$, we have from \eqref{eq:ratios} that $\frac{|\Ic_T(v\rightarrow v_i)|}{R_{v_i}}=\frac{1}{1+x_i}$ and  $x=\lambda \prod_{i\in [k_v]} \frac{1}{1+x_i}$, so by Lemma~\ref{lem:connconst} we have that
\[\bigg((R_v\Phi(R_v))^a\sum_{i\in [k_v]}\Big(\frac{|\Ic_T(v\rightarrow v_i)|}{R_{v_i}\Phi(R_{v_i})}\Big)^a\bigg)^{1/a}\leq \kappa^{1/\chi}.\]
By definition of the $d$-branching value we also have $\alpha_v=1+\frac{1}{d}\sum_{i\in [k_v]}\alpha_{v_i}\geq \frac{1}{d}\sum_{i\in [k_v]}\alpha_{v_i}$, so plugging these back into \eqref{eq:induction} yields
\[\sum_{i\in [k_v]}\Big(\frac{\alpha_{v_i}}{\alpha_v}\Big)^{1/\chi}
R_v
\Phi(R_v)\frac{|\Ic_T(v\rightarrow v_i)|}{
R_{v_i}
\Phi(R_{v_i})}\leq (d\kappa)^{1/\chi}.\]
In turn, plugging this into \eqref{eq:writeout} and using the induction hypothesis yields \eqref{eq:spectralind}, finishing the proof.
\end{proof}
\begin{remark}
For simplicity, and since it is not important for our arguments, the constant $W_k$ in the proof  depends exponentially on the degree $k$ of the root.  With a  more careful inductive proof  (cf. \cite[Proof of Lemma 14]{touniqueness}), the dependence on $k$ can be made linear. In either case, because of the high-degree vertices in $G(n,d/n)$, both bounds do not yield sufficiently strong bounds on the spectral independence of the whole distribution $\mu_{G,\lambda}$, and this is one of the reasons that we have to consider the spectral independence on the induced distribution on low-degree vertices.
\end{remark}

Recall that for a graph $G=(V,E)$ and $U\subseteq V$, we let $\mu_{G,\lambda,U}(\cdot)$ denote the marginal distribution on the spins of $U$, i.e., the distribution $\mu_{G,\lambda}(\sigma_U=\cdot)$.
\begin{lemma}\label{lem:hardspectral}
Let $d\geq 1$ and $\lambda>0$ be constants such that $\lambda<\lambda_c(d)$.   Then, for any constants $D,\epsilon>0$, whp over the choice of $G\sim G(n,d/n)$, the marginal hard-core distribution $\mu_{G,\lambda, U}$, where $U$ is the set of vertices in $G$ with degree $\leq D$, is $(\epsilon\log n)$-spectrally independent.
\end{lemma}
\begin{proof}
Let $D,\epsilon>0$ be arbitrary constants, and let $d'>d$ be such that $\lambda<\lambda_c(d')$; such $d'$ exists because the function $\lambda_c(\cdot)$ is continuous in the interval $(1,\infty)$ and $\lambda_c(d)\rightarrow \infty$ for $d\downarrow 1$.  Let $\chi\in (1,2)$ and $W=\max\{W_1,\hdots,W_D\}$ where $\chi$ and the $W_k$'s are as in Lemma~\ref{lem:spectraltrees} (corresponding to the constants $d',\lambda$). By Lemma~\ref{lem:branch}, whp all of the vertices of the graph $G=(V,E)\sim G(n,d/n)$ have $d'$-branching value less than $\epsilon \log n$. We will show that the result holds for all such graphs $G$ (for sufficiently large $n$).

Let $U$ be the set of vertices in $G$ with degree $\leq D$, and let for convenience $\mu=\mu_{G,\lambda,U}$. Consider arbitrary $S\subset U$ and $\tau\in \Omega_S$. It suffices to bound the largest eigenvalue of  the influence matrix $\Psib_\tau$ by $\epsilon \log n$. Analogously to \cite{Spectral,touniqueness}, we do this by bounding the absolute-value row sums of $\Psib_\tau$. Recall that the rows and columns of $\Psib_\tau$ are indexed by $\tilde{V}_\tau=\{(v,i)\mid  v\in U\backslash S,\,  \mu_\tau(\sigma_v=i)>0\}$, where the entry indexed by $(v,i), (w,k)$ equals $\mu_\tau(\sigma_w=k\mid \sigma_v=i)-\mu_\tau(\sigma_w=k)$ if $v\neq w$, and 0 otherwise. Consider arbitrary $(v,i)\in \tilde{V}_\tau $; our goal is to show
\begin{equation}\label{eq:goal2423432}
\sum_{(w,k)\in \tilde{V}_\tau} \big|\mu_\tau(\sigma_w=k\mid \sigma_v=i)-\mu_\tau(\sigma_w=k)\big|\leq \epsilon \log n.
\end{equation}
Henceforth, we will also assume that $\mu_\tau(\sigma_v=i)<1$ (in addition to $\mu_\tau(\sigma_v=i)>0$), otherwise the sum on the l.h.s. is equal to 0.  Then, by the law of total probability, for any $(w,k)\in \tilde{V}_\tau$  we have 
\[ \big|\mu_\tau(\sigma_w=k\mid \sigma_v=i)-\mu_\tau(\sigma_w=k)\big|\leq \big|\mu_\tau(\sigma_w=k\mid \sigma_v=1)-\mu_\tau(\sigma_w=k\mid \sigma_v=0)\big|=\big|\Ic^{\tau}_G(v\rightarrow w)\big|,\]
where the last equality follows from the fact that $\mu$ is the marginal distribution of $\mu_{G,\lambda}$ on $U$. Therefore, we can bound
\[\sum_{(w,k)\in \tilde{V}_\tau} \big|\mu_\tau(\sigma_w=k\mid \sigma_v=i)-\mu_\tau(\sigma_w=k)\big|\leq 2\sum_{w\in U}\big|\Ic^{\tau}_G(v\rightarrow w)|\leq 2\sum_{w\in V}\big|\Ic^{\tau}_G(v\rightarrow w)\big|.\]

By Lemma~\ref{lem:graphtotrees}, for the self-avoiding walk tree $T=(V_T,E_T)$ from $v$,    there is a subset $\Zeta\subseteq V_T\backslash \{\rho\}$ and a configuration $\phi\in \{0,1\}^{\Zeta}$ such that
\[2\sum_{w\in V}\big|\Ic^{\tau}_G(v\rightarrow w)\big|\leq 2\sum_{w\in V_T} \big|\Ic^{\phi}_{T}(v\rightarrow w)\big|,\]
where $\Ic^{\phi}_{T}(v\rightarrow \cdot)$ denotes the influence of $v$ on the vertices of $T$ (in the hard-core distribution $\mu_{T,\lambda}$ conditioned on $\phi$). Since the $d'$-branching value of $v$ (and any other vertex of $G$) is bounded by $\epsilon \log n$ and the degree of $v$ is $\leq D$, by Lemma~\ref{lem:spectraltrees} applied to $T$, we have that 
\[2\sum_{w\in V_T} \big|\Ic^{\phi}_{T}(v\rightarrow w)\big|\leq 2W(\epsilon \log n)^{1/\chi}.\]
Since $\chi>1$, for all sufficiently large $n$ we have that $2W(\epsilon \log n)^{1/\chi}\leq \epsilon \log n$, which proves \eqref{eq:goal2423432}.
\end{proof}
We also record the following corollary of the arguments in Lemma~\ref{lem:spectraltrees}.
\begin{corollary}\label{lem:hardbmarginal}
Let $\lambda>0$ and $D>0$ be real numbers. For a graph $G=(V,E)$, let $U$ be the set of vertices in $G$ with degree $\leq D$ and suppose that $|U|\geq 2$. Then, the distribution $\mu:=\mu_{G,\lambda,U}$ is $b$-marginally bounded for $b=\min\{\tfrac{1}{1+\lambda},\tfrac{\lambda}{\lambda+(1+\lambda)^D}\}$.
\end{corollary}
\begin{proof}
By Lemma~\ref{lem:graphtotrees}, for any vertex $v\in U$ and any  boundary condition $\tau$ on (a subset of) $U\backslash \{v\}$, there is a corresponding tree $T$ and a boundary condition $\phi$ on $T$ such that $\mu_\tau(\sigma_v=\cdot)=\nu_\phi(\sigma_v=\cdot)$. Since $v$ has degree $\leq D$, from the proof of Lemma~\ref{lem:spectraltrees}, see in particular equation \eqref{eq:ratios}, we have that $\nu_\phi(\sigma_v=\cdot)\geq b$, where $b$ is as in the lemma statement. 
\end{proof}

\section{Entropy factorisation for bounded-degree vertices}\label{sec:entropyfac}
In this section, we show how to convert the spectral independence results of the previous section into fast mixing results for Glauber dynamics on the set of small-degree vertices on $G(n,d/n)$. Our strategy here follows the technique of \cite{Optimal}, though to obtain nearly linear results we have to pay attention to the connected components induced by high-degree vertices and how these can connect up small-degree vertices.

\subsection{Preliminaries}\label{4f4f444f}

\paragraph*{Entropy factorisation for probability distributions.}
For a real function $f$ on $\Omega\subseteq [q]^V$, we use $\Eb_\mu(f)$ for the expectation of $f$ with respect to $\mu$ and, for $f:\Omega\rightarrow \mathbb{R}_{\geq 0}$,  $\Ent_\mu(f)=\Eb_{\mu}[f\log f]-\Eb_\mu(f)\log \Eb_\mu(f)$, with the convention that $0\log 0= 0$.  Finally, for $S\subset V$, let $\Ent^S_\mu(f)=\Eb_{\tau\sim \mu_{V\backslash S}}\big[\Ent_{\mu_\tau}(f)\big]$ i.e.,
$\Ent^S_\mu(f)$ is the expected value of the conditional entropy of $f$ when the assignment outside of $S$ is chosen according to the marginal distribution $\mu_{V\backslash S}$ (the induced distribution of $\mu$ on $V\backslash S$). For convenience, when $S=V$, we define $\Ent^S_\mu(f)=\Ent_\mu(f)$. The following inequality of entropy under tensor product is a special case of Shearer's inequalities.
\begin{fact}\label{fac:product}
Let $q,k\geq 2$ be integers and suppose that, for $i\in [k]$, $\mu_i$ is a distribution supported over $\Omega_i\subseteq[q]^{V_i}$, where $V_1,\hdots, V_k$ are pairwise disjoint sets. Let $\mu=\mu_1\otimes \cdots \otimes \mu_k$ be the product distribution on $\Omega=\Omega_1\times \cdots \times\Omega_k$. Then, for any $f:\Omega\rightarrow \mathbb{R}_{\geq 0}$,  it holds that $\Ent_{\mu}(f)\leq \sum_{i=1}^{k}\Ent^{V_i}_{\mu}(f)$.
\end{fact}
To bound the mixing time of Markov chains such as the Glauber dynamics, we will be interested in establishing inequalities for factorisation of entropy, defined as follows (see \cite{caputo2021} for more details).
\begin{definition}
Let $q\geq 2$, $r\geq 1$ be integers and $V$ be a set of size $n\geq r+1$. Let $\mu$ be a distribution supported over $\Omega\subseteq [q]^V$. We say that $\mu$ satisfies the $r$-uniform-block factorisation of entropy with    multiplier\footnote{We note that in related works $C_r$ is usually referred to as the ``factorisation constant''; we deviate from this terminology  since for us $C_r$ will depend on $n$ (cf. Corollary~\ref{cor:hardtensor} and Lemma~\ref{lem:hard1tensor}), and referring to it as a constant could cause confusion.} $C_r$ if for all $f:\Omega\rightarrow\mathbb{R}_{\geq 0}$ it holds that
$\frac{r}{n}\Ent_{\mu}(f)\leq C_r\frac{1}{\binom{n}{r}}\sum_{S\in \binom{V}{r}} \Ent^S_\mu(f)$.
\end{definition}
The following lemma will be useful to bound the ($r$-uniform-block) factorisation multiplier for conditional distributions on sets with small cardinality.
\begin{lemma}[{\cite[Lemma 4.2]{Optimal}}]\label{lem:crudetens}
Let $q\geq 2$ be an integer and $V$ be a set of size $n\geq 2$.  Let $\mu$ be a distribution supported over $\Omega\subseteq [q]^V$ which is $b$-marginally bounded for some $b>0$. Then, for any $S\subseteq V$ and $\tau\in \Omega_{V\backslash S}$, for $f:\Omega\rightarrow \mathbb{R}_{\geq 0}$, it holds that $\Ent_{\mu_\tau}(f)\leq \frac{2|S|^2\log(1/b)}{b^{2|S|+2}}\sum_{v\in S}\Ent^{v}_{\mu_\tau}(f)$.
\end{lemma}

\paragraph*{The $r$-uniform-block dynamics and its mixing time.}
For an integer $r=1,\hdots,n$,  the $r$-uniform-block dynamics for $\mu$ is a Markov chain $(X_t)_{t\geq 0}$ where $X_0\in \Omega$ is an arbitrary configuration and, for $t\geq 1$, $X_t$ is obtained from $X_{t-1}$ by first picking a subset $S\in V$ of size $r$ uniformly at random and updating the configuration on $S$ according to $\mu\big(\sigma_S=\cdot\mid \sigma_{V\backslash S}= X_{t-1}(V\backslash S)\big)$. Note, the case $r=1$ corresponds to the single-site dynamics, where at every step the spin of a single vertex, chosen u.a.r., is updated conditioned on the spins of the remaining vertices. For $\epsilon>0$, the mixing time of the $r$-uniform-block dynamics  is defined as $\Tmix(\epsilon)=\max_{\sigma\in \Omega}\min\big\{t\,\big|\, X_0=\sigma, \ \norm{\nu_{t}-\mu}_{\mathrm{TV}}\leq \epsilon\big\},$ where $\nu_t$ denotes the distribution of $X_t$. The following lemma builds upon a well-known connection between factorisation of entropy and  modified log-Sobolev inequalities (see, e.g., \cite{caputo2021} for more discussion), we will use the following version that can be extracted from recent works. 
\begin{lemma}[See, e.g., {\cite[Lemma 2.6 \& Fact 3.5(4)]{Optimal}} or~{\cite[Lemma 3.2.6 \& Fact 3.4.2]{zongchenphdthesis}}]\label{lem:rmixing}
Let $q\geq 2$, $r\geq 1$ be  integers and $V$ be a set of size $n\geq r+1$. Let $\mu$ be a distribution supported over $\Omega\subseteq [q]^V$ that satisfies the $r$-uniform-block  factorisation of entropy with  multiplier $C_r$. Then, for any $\epsilon>0$, the mixing time of the $r$-uniform-block dynamics  on $\mu$ satisfies
\[\Tmix(\epsilon)\leq \bigg\lceil C_r \frac{n}{r}\Big(\log\log\frac{1}{\mu_{\min}}+\log \frac{1}{2\epsilon^2}\Big)\bigg\rceil, \mbox{ where $\mu_{\min}=\min_{\sigma\in \Omega}\mu(\sigma)$.}\]
\end{lemma}
We remark that to deduce the lemma from \cite{Optimal} or \cite{zongchenphdthesis}, which refer to the so-called ``entropy decay constant $\kappa$'', one needs to use the equality $C_r\kappa=r/n$ from \cite[Lemma 2.6]{Optimal} or \cite[Lemma 3.2.6]{zongchenphdthesis}.

\paragraph*{From spectral independence to $r$-uniform-block factorisation multipliers.}
The following theorem is shown in \cite{Optimal}; while the version that we state here cannot be found verbatim in \cite{Optimal}, we explain in Appendix~\ref{sec:omitted} how to combine the results therein to obtain it.
\begin{theorem}[{\cite{Optimal}}]\label{thm:Optimal}
Let $q\geq 2$ be an integer and $V$ be a set of size $n\geq 2$.  Let $\mu$ be a distribution supported over $\Omega\subseteq [q]^V$ that is $\eta$-spectrally independent and $b$-marginally bounded for $\eta,b>0$. 

Then, for all integers $r=1,\hdots,n$, the distribution $\mu$ satisfies the $r$-uniform-block factorisation of entropy with   multiplier $C_r=\displaystyle \frac{r}{n}\frac{\sum^{n-1}_{k=0}\Gamma_k}{\sum^{n-1}_{k=n-r}\Gamma_k}$, where $\Gamma_k=\prod^{k-1}_{j=0}\alpha_j$ for $k\in [n]$\footnote{We note that for $k=0$, the product defining $\Gamma_k$ is empty and therefore $\Gamma_0=1$.}  and $\alpha_k=\max\big\{0,1-\frac{4\eta}{b^2(n-1-k)}\big\}$ for $k\in [n-1]$.
\end{theorem}

\subsection{Entropy factorisation for bounded-degree vertices in the hard-core model}\label{sec:entropyfacb}

We begin by noting that our arguments in this section, while developed primarily in the context of the hard-core model, apply more generally and will be used in particular for our Ising and monomer-dimer results on $G(n,d/n)$ (cf. Section~\ref{sec:efweef}).
 
The first step of the  analysis 
of Glauber dynamics for the hard-core model on the set of small-degree vertices will be to employ spectral independence results of Section~\ref{sec:ffferfr} to conclude fast mixing  for the $r$-uniform-block dynamics
 for  $r=\theta |U|$ for any arbitrarily small constant $\theta$. This step will follow by applying the recent technology of entropy factorisation described above. 

The second step is the more challenging for us. Here we need to conclude fast mixing for $r=1$, and in particular prove that  $C_1/C_r=n^{o(1)}$. This is done roughly by studying the connected components of $G$ that arise when resampling an $r$-subset of the low-degree vertices; the   factorisation multiplier of these components controls the ratio $C_1/C_r$.  While this resembles the approach of \cite{Optimal}, there is a key difference here, in that high-degree vertices are not resampled. This can not only cause potentially large components, but also imposes a 
deterministic lower bound on  components sizes (since a component consisting of high-degree vertices will be deterministically present in the percolated graph consisting of the $r$-subset of low-degree vertices and all of the high-degree vertices). This lower bound on the component sizes is actually more significant than it might initially seem since the relatively  straightforward bound of $\Omega(\log n)$ would unfortunately give a relatively  large factorisation multiplier of $n^{\Omega(1)}$ (through Lemma~\ref{lem:crudetens}). Instead, we need to 
show that components have size
$o(\log n)$, which in turn requires more delicate estimates for the distribution of high-degree vertices in connected sets (see Lemma~\ref{lem:percolation} below).

We start with the following corollary of Lemma~\ref{lem:hardspectral}, which converts a spectral independence bound into a bound on the  factorisation multiplier for the 
$r$-uniform-block dynamics when $r$ scales linearly with small-degree vertices. This is analogous to \cite[Lemma 2.4]{Optimal}, where they obtain a $2^{O(\eta/b^2)}$ bound on $C_r$ when $r=\Theta(n)$ via Theorem~\ref{thm:Optimal} (where $\eta$ is the spectral independence bound and $b$ is the bound on the marginals). By restricting to small-degree vertices, we obtain that $b$ is a constant, which combined with the bound $\eta=o(\log n) $ from Lemma~\ref{lem:hardspectral} gives the bound $C_r=n^{o(1)}$, as detailed below (for clarity, we show the relevant lemmas using arbitrarily small constants instead of $o(1)$, see Remark~\ref{remark} on how to extract $o(1)$). The proof of the corollary is given for completeness in Appendix~\ref{sec:omitted}.

\begin{corollary}\label{cor:hardtensor}
Let $d\geq 1$ and $\lambda>0$ be constants such that $\lambda<\lambda_c(d)$.   Then, for any constants $D,\theta>0$,  whp over the choice of $G\sim G(n,d/n)$, the marginal hard-core distribution $\mu_{G,\lambda, U}$, where $U$ is the set of vertices in $G$ with degree $\leq D$, satisfies for any integer $r\in \big[{\theta} |U|,|U|\big]$ the $r$-uniform-block factorisation of entropy with   multiplier  $C_r\leq n^\theta$.
\end{corollary}
Note that the reason that we are able to use the same $\theta$ in the bounds for both $r$ and $C_r$ is that the bound on $C_r$ is loose (we can obtain a sharper result since we have a bound on the spectral independence of $\epsilon\log n $ for any $\epsilon>0$).

We will now refine the bound of Corollary~\ref{cor:hardtensor} down to $r=1$ by exploiting the fact that high-degree vertices are sparsely scattered. In particular, we will need the following lemma which is a refinement of Lemma~\ref{lem:time}. For a graph $G=(V,E)$, we say that a set $S\subseteq V$ is connected if the induced subgraph $G[S]$ is connected.
\newcommand{\statelempercolation}{Let $d>0$ be an arbitrary real. There exists an $L>0$ such that for any $\delta\in(0,1)$,  the following holds  whp over the choice of $G=(V,E)\sim G(n,d/n)$. For $\Delta=1/(\delta \log\tfrac{1}{\delta})$, for all integers $k\geq \delta\log n$ and any $v\in V$, there are  $\leq (2\emm)^{\Delta L k}$ connected sets $S\subseteq V$ containing $v$ with  $|S|=k$. Moreover, every such set contains $\geq k/2$ vertices with degree less than $L\Delta$.}
\begin{lemma}\label{lem:percolation}
\statelempercolation
\end{lemma}

The proof of Lemma~\ref{lem:percolation} is given in Section~\ref{sec:smallcomponents}. We are now ready to show the following.
\begin{lemma}\label{lem:hard1tensor}
Let $d\geq 1$ and $\lambda>0$ be constants such that $\lambda<\lambda_c(d)$.   For any $\theta>0$, there is a constant $D>0$ such that  whp over the choice of $G\sim G(n,d/n)$, the marginal hard-core distribution $\mu_{G,\lambda, U}$, where $U$ is the set of vertices in $G$ with degree $\leq D$, satisfies the $1$-uniform-block   factorisation of entropy with   multiplier  $C_1\leq n^\theta$.
\end{lemma}
\begin{proof}
Let $L>0$ be as in Lemma~\ref{lem:percolation}, and consider an  arbitrarily small constant $\theta>0$. Let $\delta\in (0,1)$ be a sufficiently small constant so that for $D=L/(\delta \log \tfrac{1}{\delta})$ and  $b=\min\{\tfrac{1}{1+\lambda},\tfrac{\lambda}{\lambda+(1+\lambda)^{D}}\}$ it holds that $\frac{1}{b^{2\delta}}<\emm^{\theta/4}$; such a constant exists since $b^{2\delta}\rightarrow 1$ as $\delta\downarrow 0$. Let $\Delta=1/(\delta\log\tfrac{1}{\delta})$ and $\zeta>0$ be a small constant so that $2(2\emm)^{L\Delta} (2\zeta)^{1/2}\leq b^2/2$.

Let $U$ be the vertices in $G$ with degree $\leq D$, and  let $r=\lfloor\zeta|U|\rfloor+1$. Let $\mu=\mu_{G,\lambda, U}$. By Corollary~\ref{cor:hardtensor}, we have that whp over the choice of $G$, there is $C_r\leq n^{\theta/3}$ such that for every $f:\Omega\rightarrow \mathbb{R}_{\geq 0}$  it holds that
\begin{equation}\label{eq:3ed3d}
\frac{r}{|U|}\Ent_{\mu}(f)\leq C_r\frac{1}{\binom{|U|}{r}}\sum_{S\in \binom{U}{r}} \Ent^S_\mu(f).
\end{equation}

For $S\subseteq U$, let $\mathcal{C}'(S)$ denote  the collection of the connected components of the graph $G[S\cup (V\backslash U)]$, viewed as vertex sets, and let $\mathcal{C}(S)=\bigcup_{R\in \mathcal{C}'(S)}\{R\cap U\}$  be the restriction of these components to the set $U$. 
Note that, for $S\subseteq U$ and $\tau\in \Omega_{U\backslash S}$, $\mu_\tau$ factorises over the components of $G[S\cup (V\backslash U)]$
and in particular, applying Fact~\ref{fac:product}, we have that
\[\Ent^S_\mu(f)=\Eb_{\tau\sim \mu_{U\backslash S}}\big[\Ent_{\mu_\tau}(f)\big]\leq \Eb_{\tau\sim \mu_{U\backslash S}}\bigg[\sum_{R\in \mathcal{C}(S)}\Ent^{R}_{\mu_\tau}(f)\bigg].\] 
Using the bound in Lemma~\ref{lem:crudetens}, we further obtain that 
\[\Ent^S_\mu(f)\leq \Eb_{\tau\sim \mu_{U\backslash S}}\bigg[\sum_{R\in \mathcal{C}(S)} \frac{2|R|^2\log(1/b)}{b^{2|R|+2}}\sum_{u\in R}\Ent^{u}_{\mu_\tau}(f)\bigg]= \sum_{R\in \mathcal{C}(S)}\sum_{u\in R}\frac{2|R|^2\log(1/b)}{b^{2|R|+2}}\Ent^u_{\mu}(f),\]
where the last equality follows by linearity of expectation and the fact that $\Eb_{\tau\sim \mu_{U\backslash S}}[\Ent^{u}_{\mu_\tau}(f)]=\Ent^u_{\mu}(f)$. Plugging this bound into \eqref{eq:3ed3d}, we obtain that 
\[\Ent_{\mu}(f)\leq \frac{2C_r\log(1/b)}{b^2\binom{|U|-1}{r-1}}\sum_{S\in \binom{U}{r}} \sum_{R\in \mathcal{C}(S)}\sum_{u\in R}\frac{|R|^2}{b^{2|R|}}\Ent^u_{\mu}(f).\]
which yields that
\begin{equation}\label{eq:445t45t}
\Ent_{\mu}(f)\leq \frac{2C_r\log(1/b)}{b^2}\sum_{u\in U}\Ent^u_{\mu}(f)\sum^n_{k= 1}\frac{k^2}{b^{2k}}\Pr[\mathcal{C}_u(S)=k],
\end{equation}
where $\Pr[\mathcal{C}_u(S)=k]$ denotes the probability that $u$ belongs to a set of size $k$ in the set $\mathcal{C}(S)$, when we pick $S$ uniformly at random from $\{S\in \binom{U}{r}\mid u\in S\}$. Define analogously $\Pr[\mathcal{C}_u'(S)=k]$ to be the probability that $u$ belongs to a connected component of size $k$ in the set $\mathcal{C}'(S)$.  By Lemma~\ref{lem:percolation}, whp over $G\sim G(n,d/n)$, for all vertices $u$ and any integer $t\geq \delta\log n$, there are at most $(2\emm)^{L\Delta t}$ connected sets of size $t$ containing a given vertex $u$, and each of them contains at least $t/2$ vertices from $U$. In particular, for any integer $k\geq \delta \log n$, it holds that $\Pr[\mathcal{C}_u(S)=k]\leq \Pr[k\leq \mathcal{C}_u'(S)\leq 2k]$. For all $k\leq 2|U|$, the probability that a specific subset of $k/2$ vertices of $U$  is present in $G[S\cup (V\backslash U)]$ is at most $\frac{\binom{|U|-\lceil k/2\rceil}{r-\lceil k/2\rceil}}{\binom{|U|-1}{r-1}}\leq \big(\frac{r}{|U|}\big)^{k/2}\leq (2\zeta)^{k/2}$. Therefore, for all $k\geq \delta\log n$, by a union bound over the connected sets of size $k$, we have

 \[\Pr[\mathcal{C}'_u(S)=k]\leq  (2\emm)^{L\Delta k}2^k(2\zeta)^{k/2}=\big(2(2\emm)^{L\Delta} (2\zeta)^{1/2}\big)^{k}\leq (b^2/2)^k,\]
 where 
 in the first inequality the first factor is the number of size-$k$ connected sets $T$ of $G$ containing~$u$, the second factor is an upper bound on the number of size $k/2$ subsets $W$ of $U$ that might be included in $T$ and the final factor is the probability that $W$ is included in $S$. 
 The last inequality is by the choice of $\zeta$. It follows that 
 \[\Pr[\mathcal{C}_u(S)=k]\leq \Pr[k\leq \mathcal{C}'_u(S)\leq 2k]\leq 2k(b^2/2)^k.\]
 From this bound and the inequality $1/b^{2\delta}<\emm^{\theta/4}$ by the choice of $\delta$, we can split and  bound the rightmost sum in \eqref{eq:445t45t} by
 \[\sum^n_{k= 1}\frac{k^2}{b^{2k}}\Pr[\mathcal{C}_u(S)=k]\leq \frac{(\delta \log n)^2}{b^{2\delta \log n}}+\sum_{k\geq \delta \log n} \frac{2k^3}{2^k}\leq n^{\theta/3},\]
 where the last inequality holds for all sufficiently large $n$, since the first term is $O( (\log n)^2  n^{\theta/4})$ and the second term is $O(1)$. In turn, plugging this into \eqref{eq:445t45t}, we obtain that $\mu$ satisfies the 1-uniform block factorisation of entropy with   multiplier $C_1=\frac{2C_r\log(1/b)}{b^2} n^{\theta/3}\leq n^{\theta}$ for all sufficiently large $n$ (since $b$ is a constant and $C_r\leq n^{\theta/3}$), as needed.
\end{proof}

Lemma~\ref{lem:fastone} now follows easily by combining 
Lemmas~\ref{lem:rmixing} and~\ref{lem:hard1tensor}.
This was the last ingredient needed in the proof of Theorem~\ref{thm:main}.

\section{Proofs for the Ising and monomer-dimer models}

\subsection{Antiferromagnetic Ising model}
In this section, we bound the spectral independence of  $G(n,d/n)$ in the antiferromagnetic Ising model.  We will need the following well-known lemma; the version we state here is from \cite[Lemma 40]{randreg}, but it traces back to \cite{MosselSly} based on a lemma from \cite{Berger} (that was originally stated for the ferromagnetic Ising model).
\begin{lemma}[{\cite[Lemma 40]{randreg}}, see also {\cite{MosselSly,Berger}}]\label{lem:worst}
Let $\beta\in (0,1)$. Let $T=(V,E)$ be a tree, $S\subseteq V$ be a subset of the vertices and $v\in V$ be an arbitrary vertex. Let $\tau_1,\tau_2\in \{0,1\}^S$ be two configurations on $S$ which differ only a subset $W\subseteq S$. Then,
\[\big|\mu_{T,\beta}(\sigma_v=1\mid \sigma_S=\tau_1)-\mu_{T,\beta}(\sigma_v=1\mid \sigma_S=\tau_2)\big|\leq \sum_{w\in W}\Big(\frac{1-\beta}{1+\beta}\Big)^{\mathrm{dist}(v,w)},\]
where for a vertex $w\in W$, $\mathrm{dist}(v,w)$ denotes the distance between $v$ and $w$ in $T$.
\end{lemma}

\begin{lemma}\label{lem:isingspectral}
Let $d\geq 1$ and $\beta>0$ be constants such that $\beta\in (\beta_c(d),1)$.   For any constant $\epsilon>0$, whp over the choice of $G\sim G(n,d/n)$, the antiferromagnetic Ising distribution $\mu_{G,\beta}$  is $(\epsilon\log n)$-spectrally independent.
\end{lemma}
\begin{proof}
Let $d'>d$ be a real number such that $\beta>\beta_c(d')$, i.e., $\tfrac{1-\beta}{1+\beta}<\tfrac{1}{d'}$; such $d'$ exists because of the continuity of the function $\beta_c(\cdot)$ in the interval $(0,1)$. By Lemma~\ref{lem:branch}, whp all of the vertices the graph $G=(V,E)\sim G(n,d/n)$ have $d'$-branching value less than $\epsilon \log n$. We will show that the result holds for all such graphs $G$.

Let for convenience $\mu=\mu_{G,\beta}$. Consider arbitrary $S\subset V$ and $\tau\in \Omega_S$. We will once again bound the largest eigenvalue of  the influence matrix $\Psib_\tau$ by bounding the absolute-value row sums of $\Psib_\tau$. In particular, analogously to the proof of Lemma~\ref{lem:hardspectral}, it suffices to show that for $v\in V\backslash S$ we have
\begin{equation}\label{eq:rvtg56g}
\sum_{w\in V}\big|\Ic^{\tau}_G(v\rightarrow w)\big|\leq \epsilon \log n,
\end{equation}
where $\Ic^{\tau}_G(v\rightarrow w)$ denotes the influence of $v$ on the vertices of $G$ conditioned on $\tau$. By Lemma~\ref{lem:graphtotrees}, for the self-avoiding walk tree $T=(V_T,E_T)$ from $v$,    there is a subset $W\subseteq V_T\backslash \{v\}$ and a configuration $\phi\in \{0,1\}^{W}$ such that
\begin{equation}\label{eq:445f45}
\sum_{w\in V}\big|\Ic^{\tau}_G(v\rightarrow w)\big|\leq \sum_{w\in V_T} \big|\Ic^{\phi}_{T}(v\rightarrow w)\big|,
\end{equation}
where $\Ic^{\phi}_{T}(v\rightarrow \cdot)$ denotes the influence of $v$ on the vertices of $T$ (in the  Ising distribution $\mu_{T,\lambda}$ conditioned on $\phi$). Since the $d'$-branching value of $v$ (and any other vertex of $G$) is bounded by $\epsilon \log n$, and applying Lemma~\ref{lem:worst}  to $T$, we have that 
\[\sum_{w\in V_T} \big|\Ic^{\phi}_{T}(v\rightarrow w)\big|\leq \sum_{\ell\geq 0} N_{v,\ell}\Big(\frac{1-\beta}{1+\beta}\Big)^{\ell}\leq \sum_{\ell\geq 0} N_{v,\ell}/(d')^{\ell}\leq \epsilon \log n,\]
where $N_{v,\ell}$ is the number of paths in $G$ with a total of $\ell+1$ vertices starting from $v$. Combining this with \eqref{eq:445f45} proves \eqref{eq:rvtg56g}, as wanted.
\end{proof}
The following lemma is a crude (and relatively standard) bound on the marginals for the antiferromagnetic Ising distribution, a proof can be found in, e.g., \cite[Lemma 26]{randreg}; there, the result is stated for bounded-degree graphs, but the proof of the marginal bound applies to any vertex whose degree is bounded.
\begin{lemma}[see, e.g., {\cite[Lemma 26]{randreg}}]
Let $\beta\in (0,1)$ and $D>0$ be real numbers. For a graph $G=(V,E)$, let $U$ be the set of vertices in $G$ with degree $\leq D$ and suppose that $|U|\geq 2$. Then, the distribution $\mu:=\mu_{G,\beta,U}$ is $b$-marginally bounded for $b=\tfrac{\beta^D}{1+\beta^D}$.
\end{lemma}

\subsection{Monomer-dimer model}\label{sec:monomer}
In this section, we bound the spectral independence of  $G(n,d/n)$ in the monomer-dimer model. Instead of vertex-to-vertex influences that we considered for 2-spin systems, we need to consider instead edge-to-edge influences. Namely, for a graph $G=(V,E)$ and $\mu=\mu_{G,\gamma}$, fix some $\tau\in \{0,1\}^S$ for some $S\subset E$. Then, following \cite[Section 6]{Optimal},   for edges $e,f$ with $e\in E\backslash S$ and $0<\mu_\tau(\sigma_e=1)<1$, we define the influence of $e$ on $f$ (under $\mu_\tau$) as
\[\Ic^{\tau}_G(e\rightarrow f)=\mu_\tau(\sigma_f=0\mid \sigma_e=0)-\mu_\tau(\sigma_f=0\mid \sigma_e=1).\]
We will need the following technical lemma from \cite{ConnConnPaper}. The version below can be more easily derived from \cite[Lemma 15]{matchings} (where it is explained how to combine the results of \cite{ConnConnPaper}; in the notation below, the values of $a,\chi$ correspond to those of $p,q$ there, respectively).
\begin{lemma}[\cite{ConnConnPaper}]\label{lem:connconst2}
Let $d>1$ and $\gamma>0$. Let $\chi=(1+4\gamma \hat{d})^{1/2}$ where $\hat{d}=\max\{d,3/(4\gamma)\}$, and set $a=\tfrac{\chi}{\chi-1}$. Consider also the function  $\Phi(x)=\frac{1}{x(2-x)}$ for $x\in(0,1]$. Then, there is a constant $0<\kappa<1/d$ given by $\kappa=\tfrac{1}{\hat{d}}\big(1-\tfrac{2}{1+\chi}\big)^\chi$ such that the following holds for any integer $k\geq 1$.

Let $R_1,\hdots, R_k\in (0,1]$ be reals and $R=  \frac{1}{1+\gamma \sum^k_{j=1}R_j}$. Then $\big(\Phi(R)\big)^a\sum^k_{j=1}\Big(\frac{\gamma R^2}{\Phi(R_j)}\Big)^a\leq \kappa^{a/\chi}$.
\end{lemma}
With this in hand, and using the notion of branching values (cf. Definition~\ref{def:dbranchingvalue}), we are now able to bound the total sum of influences.
\begin{lemma}\label{lem:spectraltreesMonDim}
Let $d> 1$ and $\gamma>0$ be constants. Then, there is a constant  $\chi>1$ such that the following holds.

Let $T=(V,E)$ be a tree and $e$ be an edge of $T$,  whose endpoints $v_1$ and $v_2$ have $d$-branching values  $\leq \alpha$, and let $k=\max\{\mathrm{deg}(v_1),\mathrm{deg}(v_2)\}$. Then, for the monomer-dimer distribution on $T$ with parameter $\gamma$, it holds that
\[\sum_{f\in E}\big|\Ic_T(e\rightarrow f)\big|\leq W_k \alpha^{1/\chi},\]
where $W_k>0$ is a real depending only on $k$ (and the constants $d,\gamma$).
\end{lemma}
\begin{proof}
Let $\kappa\in (0,1/d)$ and $\chi>1$ be the constants from Lemma~\ref{lem:connconst2}, and  $\Phi(x)=\frac{1}{x(2-x)}$ be also as in Lemma~\ref{lem:connconst2}. Let $M_k=2k/\kappa^{1/\chi}$ and $W_k=\tfrac{M_k}{1-(d\kappa)^{1/\chi}}$.

 Let $\mu=\mu_{T,\gamma}$. For $i\in \{1,2\}$, let $\hat{T}_i$ be the subtree containing $v_i$ 
 obtained from~$T$ when we remove the edge $e$, and let $T_i$ be the tree obtained by adding the edge $e$ to $\hat{T}_i$. We will root $T_i$ at $v_{i}$, and set $\mu_i=\mu_{T_i,\gamma}$. We will show that
 \begin{equation}\label{eq:subtreei}\sum_{f\in E_{T_i}}|\Ic_{T_i}(e\rightarrow f)|\leq W_k \alpha^{1/\chi},
 \end{equation}
 where $\Ic_{T_i}(e\rightarrow \cdot)$ denotes the influence of $e$ on the edges of $T_i$ under $\mu_i$. Note that for any $f\in E_{T_i}$, we have that  $\Ic_{T}(e\rightarrow f)=\Ic_{T_i}(e\rightarrow f)$. Therefore, by adding  \eqref{eq:subtreei} for $i=1,2$, we obtain the statement of the lemma.
 
 We therefore focus on proving \eqref{eq:subtreei} for $i=1,2$. The argument is analogous to that used in Lemma~\ref{lem:spectraltrees} with suitable adaptations to account for the monomer-dimer model. We will work inductively on the depth of $T_i$.
 Analogously to that lemma, we first define for each vertex $u$ in $T_i$ the following values $\alpha_u$ and $R_u$. If $u$ is a leaf, set $\alpha_u=1$ and $R_u=1$; otherwise set $\alpha_u=1+\tfrac{1}{d}\sum^t_{j=1}\alpha_{u_j}$ and $R_u=\tfrac{1}{1+\gamma\sum^t_{j=1}R_{u_j}}$, where $u_1,\hdots, u_t$ are the children of $u$ in the rooted tree $T_i$. Note that we have that $\alpha_{v_i}\leq \alpha$, where $\alpha$ is the upper bound on the $d$-branching value of $v_i$ in the tree $T_i$. Moreover, if we denote by $T_i(u)$ the subtree of $T_i$ rooted at $u$, then a standard induction argument shows that 
\begin{equation}\label{eq:RuT}
R_u=\mu_{T_i(u),\gamma}(\mbox{$u$ is unmatched}).
\end{equation}
For an integer $h\geq 1$, let $L_i(h)$ be the set of edges at distance $h$ from the edge $e$ in the tree $T_i$, so that $L_i(1)$ consists of the edges incident to $v_i$ other than $e$, and so on. Moreover, for an edge $f$, let $v(f)$ be the endpoint of $f$ which is the farthest from $v_{3-i}$.  We will show that
\begin{equation}\label{eq:spectralind2}
\sum_{f\in L_i(h)}\Big(\frac{\alpha_{v(f)}}{\alpha_{v_i}}\Big)^{1/\chi}\frac{\big|\Ic_{T_i}(e\rightarrow f)\big|}{R_{v(f)}\Phi(R_{v(f)})}\leq M_k (d\kappa)^{h/\chi}.
\end{equation}
Note that $\alpha_v\geq 1$  and $R_v\Phi(R_v)\leq 1$ for all vertices $v$ in $T_i$,  and for the root $\alpha_{v_i}\leq \alpha$. So, \eqref{eq:spectralind2} yields $\sum_{f\in L_i(h)}\big|\Ic_{T_i}(e\rightarrow f)\big|\leq M_k\alpha^{1/\chi} (d\kappa)^{h/\chi}$ for all integer $h\geq 0$, and therefore summing over $h$, yields~\eqref{eq:subtreei} since $\kappa<1/d$. So it only remains to prove \eqref{eq:spectralind2}.

We will work inductively.  For the base case $h=1$, we have for every edge $f\in L_i(h)$ the trivial bounds $|\Ic_{T_i}(e\rightarrow f)\big|\leq 1$ and  $1/(R_{v(f)} \Phi(R_{v_f}))=2-R_{v(f)}\leq 2$. Moreover, by the definition of $\alpha_{v_i}$, we have $\alpha_{v_i}\geq \tfrac{1}{d}\alpha_{v(f)}$, so $\tfrac{\alpha_{v(f)}}{\alpha_{v_i}}\leq d$. Using these and  $M_k=2k/\kappa^{1/\chi}$, we obtain that \eqref{eq:spectralind2} holds for $h=1$.  For the induction step, consider $f\in L(h)$, and let $v=v(f)$. Suppose $v$  has $k_v\geq 0$ children in the tree $T_i$, joined by the  edges $f_j$ for $j\in [k_{v}]$. Then, for each $j\in [k_v]$, since $f$ is on the unique path joining $e$ to $f_j$ and the edges $f$, $f_j$ cannot simultaneously belong to a matching, it holds that (see \cite[Lemmas 6.11\& 6.12]{Optimal})
\begin{align*}\Ic_{T_i}(e\rightarrow f_j)&=\Ic_{T_i}(e\rightarrow f)\Ic_{T_i}(f\rightarrow f_j), \mbox{ and }\big| \Ic_{T_i}(f\rightarrow f_j)\big|=\mu_{T_i(v(f))}(\sigma_{f_j}=1)=\gamma R_{v(f_j)}R_{v(f)},
\end{align*}
where the last equality follows from \eqref{eq:RuT}. Therefore, we can write
\begin{align}
&\sum_{f\in L(h+1)}\Big(\frac{\alpha_{v(f)}}{\alpha_{v_i}}\Big)^{1/\chi}\frac{\big|\Ic_{T_i}(e\rightarrow f)\big|}{R_{v(f)}\Phi(R_{v(f)})}\notag\\
&=\sum_{f\in L(h)}\Big(\frac{\alpha_{v(f)}}{\alpha_{v_i}}\Big)^{1/\chi}\frac{|\Ic_{T_i}(e\rightarrow f)|}{R_{v(f)}\Phi(R_{v(f)})}\sum_{j\in [k_{v(f)}]}\Big(\frac{\alpha_{v(f_j)}}{\alpha_{v(f)}}\Big)^{1/\chi}R_{v(f)}\Phi(R_{v(f)})\frac{|\Ic_{T_i}(f\rightarrow f_j)|}{R_{v(f_j)}\Phi(R_{v(f_j)})},\notag\\
&\leq M_k (d\kappa)^{h/\chi}\max_{f\in L(h)}\bigg\{\sum_{j\in [k_{v(f)}]}\Big(\frac{\alpha_{v(f_j)}}{\alpha_{v(f)}}\Big)^{1/\chi}\Phi(R_{v(f)})\frac{\gamma R_{v(f)}^2}{\Phi(R_{v(f_j)})}\bigg\}.\label{eq:writeout2}
\end{align}
Consider an arbitrary $f\in L(h)$. 
Since $\tfrac{1}{\chi}+\tfrac{1}{a}=1$, by H\"{o}lder's inequality we have that
\begin{equation}\label{eq:induction2}
\sum_{j\in [k_{v(f)}]}\Big(\frac{\alpha_{v(f_j)}}{\alpha_{v(f)}}\Big)^{1/\chi}\Phi(R_{v(f)})\frac{\gamma R_{v(f)}^2}{\Phi(R_{v(f_j)})}\leq \bigg(\sum_{j\in [k_{v(f)}]}\frac{\alpha_{v(f_j)}}{\alpha_{v(f)}}\bigg)^{1/\chi}\bigg(\big(\Phi(R_{v(f)})\big)^a\sum_{i\in [k_v]}\Big(\frac{\gamma R_{v(f)}^2}{\Phi(R_{v(f_j)})}\Big)^a\bigg)^{1/a}.
\end{equation}
By Lemma~\ref{lem:connconst2} we have that $\bigg(\big(\Phi(R_{v(f)})\big)^a\sum_{i\in [k_v]}\Big(\frac{\gamma R_{v(f)}^2}{\Phi(R_{v(f_j)})}\Big)^a\bigg)^{1/a}\leq \kappa^{1/\chi}$. Moreover, by the definition of the $\alpha$ values, we have $\alpha_{v(f)}=1+\frac{1}{d}\sum_{j\in [k_{v(f)}]}\alpha_{v(f_j)}\geq \frac{1}{d}\sum_{j\in [k_{v(f)}]}\alpha_{v_{f(j)}}$. Plugging these bounds into \eqref{eq:induction2}, and then into \eqref{eq:writeout2} completes the inductive proof of \eqref{eq:spectralind2}.
\end{proof}
We can now conclude the following spectral independence property for the monomer-dimer distribution on $G(n,d/n)$.
\begin{lemma}\label{lem:monomerspectral}
Let $d\geq 1$ and $\gamma>0$. Then, for any constants $D,\epsilon>0$, whp over the choice of $G\sim G(n,d/n)$, the marginal monomer-dimer distribution $\mu_{G,\gamma, F}$, where $F$ is the set of edges in $G$ whose both endpoints have degree $\leq D$, is $(\epsilon\log n)$-spectrally independent.  
\end{lemma}
\begin{proof}
Consider arbitrary $D,\epsilon>0$ and let $d'>d$. Let $\chi>1$ and $W=\max\{W_1,\hdots,W_D\}$ where $\chi$ and the $W_k$'s are as in Lemma~\ref{lem:spectraltrees} (corresponding to the constants $d',\gamma$). By Lemma~\ref{lem:branch}, whp all of the vertices the graph $G=(V,E)\sim G(n,d/n)$ have $d'$-branching value less than $\epsilon \log n$. We will show that the result holds for all such graphs $G$.

Let for convenience $\mu=\mu_{G,\gamma, F}$, where $F$ is the set of edges whose both endpoints have degree $\leq D$. Consider arbitrary $S\subset F$ and $\tau\in \Omega_S$. We will once again bound the largest eigenvalue of  the influence matrix $\Psib_\tau$ by bounding the absolute-value row sums of $\Psib_\tau$.  It suffices to consider the case that $S$ is empty (and $\tau$ is trivial) since conditioning on an arbitrary $\tau$ is equivalent to the monomer-dimer model on a subgraph of $G$. Analogously to the hard-core model (cf. proof of Lemma~\ref{lem:hardspectral}), to bound the largest eigenvalue, it suffices to show for arbitrary $e\in F$ that 
\begin{equation}\label{eq:rvtg56a4565}
\sum_{f\in F}\big|\Ic_G(e\rightarrow f)\big|\leq \epsilon \log n,
\end{equation}
where $\Ic_G(e\rightarrow \cdot)$ denotes the influence of $e$ on the edges of $G$ under the (full) distribution $\mu_{G,\gamma}$.  In \cite[Theorem 6.2]{touniqueness}, they showed the analogue of  Lemma~\ref{lem:graphtotrees} in the case of matchings. In particular, let $v$ be one of the endpoints of $e$ and consider the self-avoiding walk tree $T=(V_T,E_T)$  emanating from $v$,   then \cite[Theorem 6.2]{Optimal} asserts that 
\begin{equation*}
\sum_{f\in E}\big|\Ic_G(e\rightarrow f)\big|\leq \sum_{f\in E_T} \big|\Ic_{T}(e\rightarrow f)\big|,
\end{equation*}
where $\Ic_{T}(e\rightarrow \cdot)$ denotes the influence of $e$ to the edges of $T$, under $\mu_{T,\gamma}$. By Lemma~\ref{lem:spectraltreesMonDim} applied to the tree $T$, the right-hand side is bounded by $W(\epsilon \log n)^{1/\chi}$, which is less than $\epsilon \log n$ for all sufficiently large $n$. This gives \eqref{eq:rvtg56a4565}, as wanted. 
\end{proof}

\subsection{Proof of Theorems~\ref{thm:main1} and~\ref{thm:main2}}\label{sec:efweef}

\begin{proof}[Proof of Theorem~\ref{thm:main1}]
This is completely analogous to the proof of Theorem~\ref{thm:main} presented in Section~\ref{sec:outline}. Lemma~\ref{lem:time} is about structural properties of the random graph $G\sim G(n,d/n)$, so we can use it verbatim. Therefore, we only need to establish the analogue of Lemma~\ref{lem:fastone}, i.e., for any $\theta>0$ and all sufficiently large $D>0$, whp over $G\sim G(n,d/n)$, running Glauber dynamics on $U$, where $U$ is the set of vertices with degree $\leq D$, gives an $\epsilon$-sample from $\mu_{G,\beta,U}$ in time $n^{1+\theta}\log\tfrac{1}{\epsilon}$. Completely analogously to the hard-core model, using now the spectral independence bound of Lemma~\ref{lem:isingspectral}, we conclude that for any integer $r\in [\theta |U|,|U|]$ the distribution $\mu_{G,\beta,U}$ satisfies the $r$-uniform-block factorisation of entropy with   multiplier $C_r\leq n^{\theta/3}$ (i.e., the analogue of Corollary~\ref{cor:hardtensor} for the antiferromagnetic Ising model). From there, the same argument as in Lemma~\ref{lem:hard1tensor} yields  that $\mu_{G,\beta,U}$ satisfies the 1-uniform-block factorisation of entropy with    multiplier $C_1\leq n^\theta$.
\end{proof}

\begin{proof}[Proof of Theorem~\ref{thm:main2}]\label{sec:remainingproofs}
Again, this is very similar to the proof of Theorems~\ref{thm:main} and~\ref{thm:main1}, we only need to adapt the argument slightly to account for updating edges (instead of vertices).

Once again, we only need to establish the analogue of Lemma~\ref{lem:fastone}, i.e., for any $\theta>0$ and all sufficiently large $D>0$, whp over $G\sim G(n,d/n)$, running Glauber dynamics on $F$, where $F$ is the set of vertices whose both endpoints have degree $\leq D$, gives an $\epsilon$-sample from $\mu_{G,\gamma,F}$ in time $n^{1+\theta}\log\tfrac{1}{\epsilon}$. Using  the spectral independence bound of Lemma~\ref{lem:monomerspectral}, we obtain the analogue of Corollary~\ref{cor:hardtensor}, i.e.,   for any integer $r\in [\theta |F|,|F|]$ the distribution $\mu_{G,\beta,F}$ satisfies the $r$-uniform-block factorisation of entropy with multiplier $C_r\leq n^{\theta/3}$. The same argument as in Lemma~\ref{lem:hard1tensor} yields  that $\mu_{G,\beta,F}$ satisfies the 1-uniform-block factorisation of entropy with multiplier $C_1\leq n^\theta$.
\end{proof}

\section{Random Graph Properties used in our Proofs}\label{sec:randomproofs}
In this section, we prove the random graph properties that we have used in our proofs, i.e., Lemmas~\ref{lem:time},~\ref{lem:branch} and~\ref{lem:percolation}.

\subsection{Bounding the branching value}\label{sec:branch}
Let $d\geq 1$. In this section we prove Lemma~\ref{lem:branch} which bounds the $d$-branching values of the vertices in $G(n,d/n)$. The key ingredient is to bound the following closely related quantity. For a graph $G$ and a vertex $v$ of $G$, let $\hat{N}_{v,r}$ be the number of vertices at (graph) distance exactly $r$ from $v$. We call a vertex $v$ $\epsilon$-good if it holds that 
\[\hat{S}_v:=\sum_{r\geq 0} \hat{N}_{v,r}/((1+\epsilon)d)^r\leq \epsilon \log n.\]
Note that the main difference between $\hat{S}_v$ and the $d$-branching value $S_v$ (cf. Definition~\ref{def:dbranchingvalue}) is that the latter is defined with respect to the number of paths, and therefore we trivially have $\hat{S}_v\leq S_v$. For random graphs however we will also be able to show that $S_v\leq 2\hat{S}_v$. 

We start with showing that whp all vertices in $G(n,d/n)$ are $\epsilon$-good.
\begin{lemma}\label{lem:good}
For any constants $d\geq 1$ and $\epsilon>0$, whp over the choice of  $G\sim G(n,d/n)$, all vertices of $G$ are $\epsilon$-good.
\end{lemma}
\begin{proof}
For integers $n\geq 1$, consider a (finite) tree rooted at $\rho$ created by a branching process where a node at depth $r$ has $\mathrm{Bin}(n-r,\tfrac{d}{n-r})$ children, and let $Y_{n,r}$ be the number of nodes at depth $r$. Set
\[
X_n = \sum_{r=0}^{n-1} \alpha^r\, Y_{n,r}\mbox{ where } \alpha:=\tfrac{1}{d(1+\epsilon)}<1/d\leq 1.
\]
For any fixed vertex $v$ of $G$, we can couple $S_v$ and $X_n$, so that $S_v\leq X_n$. Indeed, the BFS exploration process in $G$ starting from $v$ is stochastically dominated by the branching process above, since at depth $r$ of the BFS process the number of unexplored vertices is at most $n-r$ and the connection probability $\tfrac{d}{n}$ is less than $\tfrac{d}{n-r}$. In particular, there is a coupling so that $\hat{N}_{v,r}\leq Y_{n,r}$ for all integers $r\geq 0$, therefore yielding $\hat{S}_v\leq X_n$ as well. It follows that 
\begin{equation}\label{eq:domination}
\Pr[\hat{S}_v\geq \epsilon \log n]\leq \Pr[X_n\geq \epsilon \log n].
\end{equation}

To bound the probability $\Pr[X_n\geq \epsilon \log n]$ we will use moment generating functions, adapting an argument from \cite{MosselSly}. For integer $n\geq 1$ and real $t\geq 0$, let $g_n(t)=\Eb[\exp(tX_n)]$ and consider the functions $\{f_n(t)\}_{n\geq 1}$ defined for $t\geq 0$ by
\[f_1(t)=\exp(t),\quad  f_n(t)=\exp\big(t+d(f_{n-1}(\alpha t)-1)\big)  \mbox{ for } n\geq 2.\]
We will show by induction on $n$ that $g_n(t)\leq f_n(t)$ for all $t\geq 0$. Since $X_1=1$, the base case is trivial. For $n\geq 2$, observe that 
\[X_n=1+\alpha\big(X^{(1)}_{n-1}+\hdots X^{(K)}_{n-1}\big) \mbox{ where } K\sim \mathrm{Bin}(n,\tfrac{d}{n}),\]
and $X^{(1)}_{n-1},X^{(2)}_{n-1},\hdots$ are i.i.d. variables with distribution $X_{n-1}$. Therefore,
\begin{equation*}
\begin{aligned}
g_n(t)&= \sum_{k=0}^n \binom{n}{k}\Big(\frac{d}{n}\Big)^k\Big(1-\frac{d}{n}\Big)^k \Eb\big[\emm^{t (1+ \alpha (X_{n-1}^{(1)} + \cdots + X_{n-1}^{(k)})}\big]=\emm^{t}\Big(1+\frac{d}{n}(g_{n-1}(t\alpha)-1)\Big)^n\\
&\leq \emm^{t}\Big(1+\frac{d}{n}(f_{n-1}(t\alpha)-1)\Big)^n\leq 
\exp \big(t + d (f_{n-1}(\alpha t)-1)\big)=f_n(t),
\end{aligned}
\end{equation*}
completing the inductive proof that $g_n(t)\leq f_n(t)$.

We will next define a function $F(t)$ so that $f_n(t)\leq F(t)$ for all $n$ and $t$. To define $F$,  let $c$ be a constant such that $c > 1/(1-\alpha d)$ and note that $c>1$ since $\alpha d<1$. Let $t_0>0$ be such that  
\begin{equation}\label{eq:deft0}
\exp(c\alpha t_0) = 1 + \frac{c-1} d t_0,
\end{equation}
and note that such $t_0$ exists, since both sides at $t=0$ have value $1$, and their derivatives at $t=0$ are given by  $c\alpha$ and $(c-1)/d$, satisfying  
 $c\alpha<(c-1)/d$ from the choice of $c$. From convexity of $\exp(c\alpha t)$, we obtain that for all $t\in [0,t_0]$ it holds that
\begin{equation}\label{ea2}
\exp(c\alpha t) \leq 1 + \frac{c-1} d t, \mbox{ or  equivalently } \exp\big(t + d (\exp(c \alpha t) - 1)\big)\leq \exp(ct).
\end{equation}
Now, consider the function $F$ defined by setting $F(t)=\exp(ct)$ for $t\in [0,t_0]$; for $t>t_0$ we define $F$ inductively on intervals $(t_0 (1/\alpha)^{k},t_0 (1/\alpha)^{k+1}]$
by
\[F(t) = \exp\big(t + d (F(\alpha t) - 1)\big).\]
It is not hard to see that the function $F$ is increasing and continuous, satisfying the  inequality
\begin{equation}\label{e1}
\exp(t + d (F(\alpha t) - 1))\leq F(t).
\end{equation} 
Using \eqref{ea2}, the validity of \eqref{e1} and monotonicity are trivial; to check continuity, note that $F$ can only be discontinuous at the points $t_k=t_0(1/\alpha^k)$; let $k\geq 0$ be the first such integer. By construction,  $F$ is continuous at $t_0$ (see \eqref{eq:deft0}), and hence it must be the case that $k\geq 1$. But then
\[\lim_{t\downarrow t_k }F(t)=\lim_{t\downarrow t_k}\exp\big(t + d (F(\alpha t) - 1)\big)=\exp\big(t + d (F(t_{k-1}) - 1)\big)=F(t_k),\]
contradicting the choice of $k$. Now we also prove, by induction, that $f_n(t) \leq F(t)$. We have $f_1(t) = \exp(t) \leq F(t)$, since $c > 1/(1-\alpha d) > 1$. Now suppose $f_n(t)\leq F(t)$. Then, using~\eqref{e1}, we have
\[
f_{n+1}(t) = \exp(t-d + d f_n(\alpha t))\leq \exp(t-d + d F(\alpha t))\leq F(t),\]
as needed.

Using Markov's inequality and the fact that $g_n(t)\leq f_n(t)\leq F(t)$ for arbitrary $t\geq 0$ we have 
\begin{equation}\label{eq:Markov}
\Pr(X_n\geq \epsilon \log n) = \Pr\big(\exp(tX_n)\geq\exp(t\epsilon \log n)\big)\leq \frac{F(t)}{\exp(t\epsilon \log n)}\leq \exp\big(-h(\epsilon \log n)\big),
\end{equation}
where  the function $h$ is defined for $A\geq 0$ by
\[h(A) := \sup_{t\geq 0}\, \{t A - \log F(t)\}.\]
Let $M=1+5/\epsilon$. For any $A\geq F(M)$,  we have 
\[
\frac{h(A)}{A} = \sup_{t\geq 0} \Big\{t - \frac{\log F(t)}{A}\Big\}\geq M-1= 5/\epsilon,\]
where the inequality follows by considering $t=M$ and observing that $\log F(M)\leq F(M)\leq A$.

Since $F(M)$ is a constant (depending only on $\epsilon$, $d$), we have that for all sufficiently large $n$ it holds that $\epsilon \log n\geq F(M)$, and therefore $h(\epsilon \log n)\geq 5\log n$. From \eqref{eq:domination} and \eqref{eq:Markov}, we therefore obtain that $\Pr[S_v\geq \epsilon \log n]\leq 1/n^{5}$, and Lemma~\ref{lem:good} follows by taking a union bound over the $n$ vertices of~$G$.
\end{proof}
To obtain the desired bound on the branching values from Lemma~\ref{lem:good}, we will need to relate $N_{v,r}$, the number of paths with $r+1$ vertices from $v$, with $\hat{N}_{v,r}$, the number of nodes at distance $r$ from $v$. This will follow by the following tree-like property of $G(n,d/n)$. For a graph $G$, a vertex $v$ of $G$ and real $r>0$, we use $B(v,r)$ to be the set of vertices in $G$ at distance $\leq r$ from $v$.
\begin{lemma}[{\cite[Lemma 7]{MosselSly}}]\label{lem:treelike}
Let $d\geq 1$. The following holds whp over $G=(V,E)\sim G(n,d/n)$ and $R=(\log \log n)^2$. For all $v\in V$, $|B(v,R)|\leq d^R \log n$ and the tree-excess of the induced graph $G[B(v,R)]$ is at most~1.
\end{lemma}
We can now prove Lemma~\ref{lem:branch}, which we restate here for convenience.
\begin{lembranch}
\statelembranch
\end{lembranch}
\begin{proof}
It suffices to prove the result for arbitrarily small constant $\epsilon>0$ and $d'=(1+\epsilon) d$.  Let $R=\lfloor (\log \log n)^2\rfloor-1$ and $\epsilon'=\frac{d'-1}{4d'}\epsilon$. We have that whp $G=(V,E)\sim G(n,d/n)$ satisfies the conclusions of Lemmas~\ref{lem:good} and~\ref{lem:treelike}.

For arbitrary $v\in V$, we can bound the $d'$-branching value of $v$ from above by $\sum_{i\geq 0}s_i$, where for an integer $i\geq 0$, we set $s_i:=\sum^{ (i+1)R}_{r= iR}N_{v,r}/(d')^r$. From Lemma~\ref{lem:treelike}, the tree-excess of $G[B(v,R+1)]$ is at most 1, and hence for integers $r\leq R+1$ we have the bound\footnote{This can be proved by considering the BFS tree from $v$. An ``excess'' edge can only connect vertices either at the same or consecutive  levels of the BFS tree. In either case,  every path of length $r$ that uses that edge can be mapped injectively to a vertex at distance $\leq r$ from $v$.} 
\begin{equation}\label{eq:Nvr}
N_{v,r}\leq 2\sum^r_{r'=0}\hat{N}_{v,r'}.
\end{equation}
It follows that $s_0\leq 2\sum^{R}_{r'=0}\frac{\hat{N}_{v,r'}}{(d')^{r'}}\sum^{ R}_{r=r'}\frac{1}{(d')^{r-r'}}\leq \frac{\epsilon}{2} \log n$, 
where in the last inequality we used the fact from Lemma~\ref{lem:good} that all vertices in $G$ are $\epsilon'$-good, and in particular that $\sum^{R}_{r'=0}\frac{\hat{N}_{v,r'}}{(d')^{r'}}\leq \epsilon' \log n$.  

To bound $s_i$ for an integer $i\geq 1$, note that for integer $i R\leq r\leq (i+1)R$, we can decompose a path with $r$ vertices starting from $v$ into two paths with $r-R$ and $R+1$ vertices, so we have 
\[N_{v,r}\leq N_{v,r-R}\max_{w\in B(v,r-R)} N_{w,R+1}\leq 2 d^{R+1} N_{v,r-R}\log n ,\]
where the last inequality follows from applying \eqref{eq:Nvr} for $r=R+1$ and noting that  $\sum^{R+1}_{r'=0}\hat{N}_{v,r'}=|B(v,R+1)|\leq d^{R+1}\log n$ from Lemma~\ref{lem:treelike}. Therefore
\begin{equation}\label{eq:erf4f5345}
s_i=\sum^{(i+1) R}_{r=iR}\frac{N_{v,r}}{(d')^r}\leq s_{i-1}\sum^{(i+1)R}_{r=iR}\frac{2d^{R+1}\log n}{(d')^{R}}\leq s_{i-1}\frac{2(R+1)d^{R+1}\log n}{(d')^{R}}\leq s_{i-1}/2,
\end{equation}
where the last inequality is true for all sufficiently large $n$ since $d'>d$ and $R=\omega(\log \log n)$.

Using \eqref{eq:erf4f5345} and summing over $i\geq 0$, we have that the $d'$-branching value of $v$ is bounded by $\sum_{i\geq 0} s_i\leq 2s_0\leq \epsilon \log n$. Since $v$ was an arbitrary vertex of $G$, this finishes the proof.  
\end{proof}
\subsection{Bounding the tree-excess of small connected sets} 
We start with the following lemma which shows that the tree-excess of logarithmically-sized connected sets in $G(n,d/n)$ is bounded by an absolute constant (a similar result in a slightly different setting was shown in \cite{JamesRANDOM}).
\begin{lemma}\label{lem:next}
Let $d,M>0$ be arbitrary reals. There exists a positive integer $\ell$ such that the following holds whp over the choice of $G=G(n,d/n)$. There is no connected set $S$ of vertices such that $|S|\leq M\log n$ and the tree excess is more than $\ell$. 
\end{lemma}
\begin{proof}
It suffices to prove the lemma for all sufficiently large $d$ since the property is increasing under edge inclusion; in particular we will assume that $d\geq \emm$. Let $\ell$ be an integer bigger than $2M+5$.

For a positive integer $k\leq \left\lfloor M\log n\right\rfloor$, we calculate the expected number of  sets $S$ as in the statement of the lemma with $|S|=k$. There are  $\binom{n}{k}$ choices for the set $S$, $k^{k-2}$ labelled trees on $S$, and at most $\binom{k^2}{\ell}$ choices for $\ell$ additional edges. The probability that all these $k+\ell-1$ edges are in~$G$
is $(d/n)^{k-1+\ell}$. Therefore, by a union bound, we can upper bound the expected number of such sets $S$ with $|S|=k$ by 
\[\binom{n}{k} k^{k-2}  \binom{k^2}{\ell} \left(\frac{d}{n}\right)^{k+\ell-1}
\leq \frac{\emm^{k+\ell} d^{k+\ell-1}k^{2\ell}}{n^{\ell-1} \ell^\ell} \leq
\frac{1}{n^2},\]
where the last inequality holds for all sufficiently large $n$ using that $\ell>2M+5$. By summing over the $O(\log n)$ values for the integer $k$, we obtain that the expected number of such sets $S$ with $|S|\leq M\log n$ is $o(1)$, and therefore the result follows by Markov's inequality.
\end{proof}
 \subsection{Bounding the number of high-degree vertices in connected components}
\label{sec:smallcomponents}

In this section, we prove Lemma~\ref{lem:percolation}, that bounds the number of connected sets in $G(n,d/n)$ as well as the fraction of high-degree vertices. 
(Lemma~\ref{lem:percolation} shows the existence of a quantity~$L$ which depends on a quantity $M$ from Lemma~\ref{lem:FR}: we take 
$\delta$ to be a sufficiently small constant,
$\Delta = 1/(\delta \log \tfrac1\delta)$,  
$M=\max\{10d,50(1+\log d)\}$, and $L = 4M$. We state these quantities upfront to help with the readability of the expressions in this section.)

For a graph $G=(V,E)$ and a set $S\subseteq V$, we let $\mathrm{deg}_G(S)$ denote the sum of the degrees of
the vertices in~$S$. To control the number of connected sets containing a given vertex, we will use the following bound that holds for general graphs $G$.
\begin{lemma}[{\cite[Lemma 6]{JamesRANDOM}}]\label{lem:connected}
Let $G = (V, E)$ be a graph, $v\in V$, and $\ell\geq 1$ be an integer. The number of
connected sets $S\subseteq V$ such that $v\in V$ and $\mathrm{deg}_G(S) = \ell$ is at most $(2\emm)^{2\ell-1}$.
\end{lemma} We will use the following version of Chernoff bounds to upper bound the total degree of connected sets with $\Omega(\log n)$ size. 
\begin{lemma}\label{lem:chernoff}
Suppose $X_1, \hdots, X_n$ are i.i.d. Bernoulli random variables. Let $X=X_1+\cdots+X_n$ denote their sum and set $\mu=\Eb[X]$.  Then, for any $t\geq 5$, $\Pr(X>t\mu)\leq \emm^{-\tfrac{1}{2}t\mu\log t}$.
\end{lemma}
\begin{proof}
It is well-known (see, e.g., \cite[Theorem 2.1]{Jansonbook}) that for any $\delta>0$, it holds that
\[\Pr(X>(1+\delta)\mu)\leq \Big(\frac{\emm^\delta}{(1+\delta)^{(1+\delta)}}\Big)^{\mu},\]
For $\delta\geq 4$, we have that $\frac{(1+\delta)^{(1+\delta)}}{\emm^\delta}\geq (1+\delta)^{(1+\delta)/2}$, and the result follows.
\end{proof}
\begin{lemma}\label{lem:FR}
Let $d>0$ be an arbitrary real. There is an $M>0$ such that for any $\delta\in (0,1)$, whp over
the choice of $G=G(n,d/n)$, every connected set $S$ of
vertices with $|S|=\lceil\delta \log n\rceil $ satisfies
$\mathrm{deg}_G(S) \leq M\Delta|S|$, where $\Delta=1/(\delta\log\tfrac{1}{\delta})$.
\end{lemma} 
\begin{proof}
The proof is close to an argument of  Fountoulakis and Reed \cite[Lemma 2.4]{FR}, the only difference is that we have to account for the smaller size of $S$. We may assume for convenience that $d\geq 1$ since the graph property is increasing with respect to adding edges in the graph $G$. Let $M=\max\{10d,50(1+\log d)\}$, and $\delta\in (0,1)$ be an arbitrarily small constant, we will in fact assume henceforth that $\delta<1/\emm^2$. Let $\Delta=1/(\delta \log(1/\delta))$ and note that $\Delta>2$. For convenience, let $k=\lceil \delta \log n\rceil$.

For a  set $S$, let $e_{\mathsf{in}}(S)$ be the number of edges whose both endpoints belong to $S$, and $e_{\mathsf{out}}(S)$ be the number of edges with exactly one endpoint in $S$. By Lemma~\ref{lem:next}, whp we conclude the crude bound $e_{\mathsf{in}}(S)\leq 2|S|= 2k$ for all connected sets $S$ with size $k$ (since for such sets $S$ the tree excess of $G[S]$ is bounded by an absolute constant and $|S|=k=\Omega(\log n)$).   We also have that $\Eb[e_{\mathsf{out}}(S)]\leq d k$, so by Lemma~\ref{lem:chernoff},
\[\Pr\big(e_{\mathsf{out}}(S)\geq \tfrac{1}{2} k M\Delta\big)\leq \emm^{ -\tfrac{1}{4} k M\Delta\log \tfrac{\Delta M}{4d}}\]
Since $M\geq 4d, \delta<1/\emm^2$ and $\Delta=1/(\delta \log \tfrac{1}{\delta})$, we have that $\log \tfrac{ \Delta M}{4d}\geq \log \Delta\geq \tfrac{1}{2}\log(1/\delta)$. Therefore, since  $k\geq \delta \log n$ and  $M\geq 50(1+\log d)$, we have that
\[\tfrac{1}{4} k M\Delta\log \tfrac{ M\Delta}{4d}\geq 3(1+\log d)\log n.\]
Using that the number of labelled trees on a vertex set of size $k$ is $k^{k-2}$, we have that the expected number of trees  with size $k$ in $G(n,d/n)$  is $\binom{n}{k}k^{k-2}(d/n)^{k}\leq n(\emm d)^{k}$, and hence there are in expectation at most $n(\emm d)^k$ connected sets $S$ with size $k$. Note also that the event $e_{\mathsf{out}}(S)\geq \tfrac{1}{2} k M\Delta$ is independent of the event that $S$ is connected. Therefore, by linearity of expectation we obtain that the number of   connected sets $S$ with $|S|=k$ and $e_{\mathsf{out}}(S)\geq \tfrac{1}{2} k M\Delta$ is at most
$n(\emm d)^k\emm^{-3(1+\log d)\log n}\leq n (\emm d)^{\log n}\emm^{-3(1+\log d)\log n}\leq 1/n$,
and hence whp there are no such sets by Markov's inequality.
This finishes the proof of the lemma since any 
other connected set $S$ of size $k$ satisfies 
$\deg_G(S) \leq 2 e_{\mathsf{in}}(S) + e_{\mathsf{out}}(S) \leq 4 k + \tfrac12 k M \Delta \leq M \Delta k$.
\end{proof}
\begin{lempercolation}
\statelempercolation
\end{lempercolation}
\begin{proof}
Let $M$ be as in Lemma~\ref{lem:FR}, and let $L=4M$. Consider arbitrary $\delta\in(0,1)$.  By Lemma~\ref{lem:FR}, we have whp that 
every connected set  $S$  of vertices in $G$ with $|S|= r:=\lceil\delta\log n\rceil$ has in total at most $M\Delta|S|$ edges incident to it, where $\Delta=1/(\delta\log\tfrac{1}{\delta})$. We claim that for all connected sets $S$ with $|S|\geq r$, it holds that $\mathrm{deg}_G(S)\leq 2M\Delta|S|$. Indeed,  we can decompose any connected set $S$ into at most $t\leq 2|S|/r$ connected sets (not necessarily disjoint) $S_1,\hdots, S_t$, each of size $k$.\footnote{One way to do this is to consider a spanning tree of $S$, double its edges, and obtain an Eulerian tour of the resulting graph. The desired decomposition of $S$ can then be obtained by traversing the tour and extracting connected sets with $k$ vertices.} From this, it follows that \[\mathrm{deg}_G(S)\leq \sum^t_{i=1}\mathrm{deg}_G(S_i)\leq t M\Delta r\leq 2M\Delta|S|,\]
as claimed. The bound on the number connected sets $S$ with $|S|=k\geq \delta\log n$ containing $v$ follows by applying Lemma~\ref{lem:connected}, by  aggregating over the possible values of $\deg_G(S)$, which can be at most $2M\Delta k$. Moreover, the number of vertices in such a set $S$ with degree $\geq L\Delta$ has to be at most $k/2$ (otherwise, $\deg_G(S)> L\Delta (k/2)=2M\Delta k$). This finishes the proof.
\end{proof}

\begin{lemtime}
\statelemtime
\end{lemtime}
\begin{proof}
 Let $D=\emm L$, where $L>0$ is the constant from Lemma~\ref{lem:percolation}. Let $\ell$ be the integer of Lemma~\ref{lem:next} corresponding to $M=1/\emm$.
 
 By Lemma~\ref{lem:percolation} applied to $\delta=1/\emm$, whp over the choice of $G$,  every connected set $S$ with size $k\geq \tfrac{1}{\emm}\log n$ has at least $k/2$ vertices with degree $\leq D=\emm L$. Therefore, with $U$ being the set of vertices with degree $\leq D$, we obtain that the components of $G[V\backslash U]$ have size at most $\tfrac{1}{\emm}\log n$. Moreover, by Lemma~\ref{lem:next}, whp over the choice of $G$, all these components have tree excess $\leq \ell$.
\end{proof}

\subsection{Verifying the random graph properties efficiently}\label{sec:verify}
Here, we briefly discuss how to verify in time $n^{1+o(1)}$ that a random graph $G\sim G(n,d/n)$ satisfies the high-probability properties of Lemmas~\ref{lem:time},~\ref{lem:branch} and~\ref{lem:percolation}. The property in Lemma~\ref{lem:time} is immediate since we only need to do an exploration of the graph $G[V\backslash U]$. To ensure that the property of Lemma~\ref{lem:branch} is satisfied, and following its proof, we only need to check the property in Lemma~\ref{lem:treelike}, and that for $R=\lfloor (\log \log n)^2\rfloor-1$ and all $v\in V$ it holds that  $\sum^{R}_{r=0}N_{v,r}/d^r=o(\log n)$. Both of these properties can be checked via enumeration in time $n^{1+o(1)}$. Similarly for the property in Lemma~\ref{lem:percolation}, whose proof used only the property in Lemma~\ref{lem:FR} (which is verifiable in time $n^{1+o(1)}$).

\begin{remark}\label{remark}
We claimed in the Introduction that
the family of $O(n^{1+\theta})$ algorithms from Theorem~\ref{thm:main} can be turned into an $n^{1+o(1)}$ algorithm. We now explain how to do this for the interested reader.

As noted after the statement of Theorem~\ref{thm:main},
there is a function $f_{d,\lambda,\theta}:{\mathbb{Z}}\rightarrow {\mathbb{R}}$ such that $\lim_{n\rightarrow\infty} f_{d,\lambda,\theta}(n) = 0$ and the ``whp'' 
bound in Theorems~\ref{thm:main}, \ref{thm:main1}, and \ref{thm:main2}
means with probability $\geq 1 - f_{d,\lambda,\theta}(n)$; the function equals 1 for small $n$ (making the conclusion trivial for such $n$). 

To understand the function $f_{d,\lambda,\theta}(n)$, 
we need to look at the whp bounds in the lemmas that we use.
The whp bound functions are $\geq 1-1/n$ for sufficiently large $n$ where
sufficiently large is (ultimately) a simple function of $d,\lambda, \theta$ (for example, Lemma~\ref{lem:time},~\ref{lem:percolation}, \ref{lem:next}, and~\ref{lem:FR}). Lemma~\ref{lem:good} (and as a consequence Lemma~\ref{lem:branch}) 
has a more wild dependence on the parameter  $\varepsilon$: the function $F(t)$ inside the proof can be bounded from above by a tower of exponentials of depth $log_d(t)$. The whp bound is then $\geq 1-1/n$ assuming 
$n \geq \exp ( F(1+5/\varepsilon) / \varepsilon )$; this whp bound propagates to Lemma~\ref{lem:hard1tensor}, then 
to Lemma~\ref{lem:fastone} and finally to Theorems~\ref{thm:main}, \ref{thm:main1}, and \ref{thm:main2}.

For fixed $d$ and $\lambda$, the dependence of $\varepsilon$ on $\theta$ (used in Lemma~\ref{lem:good})  can be extracted from Corollary \ref{cor:hardtensor} and Lemma~\ref{lem:hard1tensor}, where $\epsilon$ (controlling the leading constant in the branching value) needs to be small enough to compensate the leading constant from the spectral-independence bound $\eta$; in particular, $\epsilon$ scales roughly as $1/\emm^{\emm^{C/\theta}}$, where $C$ is a constant depending only on $d$ and $\lambda$. It follows that we can set $f_{d,\lambda,\theta}(n)=1/n$ when $n \geq N_{d,\lambda}(\theta)$ and 1 otherwise, where $N_{d,\lambda}(\theta)$ is a computable function satisfying $n\geq N_{d,\lambda}(1/k)$ for some function $k=k(n)=\omega(1)$ (for example, it suffices to take $k(n)=\lfloor \hat{C}\log^{(3)}\log^*\log^{(2)} n)\rfloor$ where $\hat{C}$ is a constant depending only on $d$ and $\lambda$).

The $n^{1+o(1)}$ algorithm then proceeds as follows. Given the input size~$n$,
it first computes  $k=k(n)$ as above in $O(\log n)$ time (say). Then, by the definition of $k(n)$, we have $n\geq N_{d,\lambda}(1/k)$ for all sufficiently large $n$. The new algorithm runs the algorithm of Theorem~\ref{thm:main} with $\theta=1/k$
in time $n^{1+1/k} \log \tfrac{1}{\epsilon} = 
n^{1+o(1)} \log \tfrac{1}{\epsilon}$, which  succeeds with probability $\geq 1-1/n=1-o(1)$ over the choice of $G(n,d/n)$. 
\end{remark}
 

\bibliography{references}

\newpage

\appendix

\section{Other omitted proofs}\label{sec:omitted}

\begin{proof}[Proof of Theorem~\ref{thm:Optimal}]
The theorem follows by combining \cite[Claim 1.13, Theorem 1.14, Lemma 2.6]{Optimal}, which bound the factorisation multiplier of the more general order-$(r,s)$ down-up walk; the result here is the special case $s=n$. The only difference  is that in \cite{Optimal} they state the $r$-uniform-block multiplier $\hat{C}_r=\frac{r}{n}\frac{\sum^{n-1}_{k=0}\hat{\Gamma}_k}{\sum^{n-1}_{k=n-r}\hat{\Gamma}_k}$, with $\hat{\Gamma}_k=\prod^{k-1}_{j=0}\hat{\alpha}_j$ where 
\[\hat{\alpha}_k= \max\Big\{1-\frac{4\eta_k}{b^2},\frac{1-\eta_k}{4+2\log(\frac{1}{2b_kb_{k+1}})}\Big\},\]
with $\eta_k=\tfrac{\eta}{n-1-k}$ for $k\in [n-1]$, $b_k=\tfrac{b}{n-k}$ for $k\in [n]$. The only thing we need to note, which is also implicit in \cite[Proof of Lemma~2.4]{Optimal}, is that $\hat{C}_r\leq C_r$. In particular, for each $k\in [n]$,  $C_r$ is decreasing with respect to $\alpha_k$ (both numerator and denominator are multi-linear functions and the value of $C_r$ increases as $\alpha_k\downarrow 0$), so we only need to check that  $\hat{\alpha}_k\geq \alpha_k$. For this, we need to further note that $\eta_k$ is an upper bound to the so-called local spectral expansion which is defined as the second largest eigenvalue of the transition matrix of an appropriate random walk (on a suitable simplicial complex corresponding to $\mu$; the details are not important for our purposes). This implies that $\eta_k\leq 1$, and hence $\hat{\alpha}_k\geq 0$ since $\hat{\alpha}_k\geq \frac{1-\eta_k}{4+2\log(\frac{1}{2b_kb_{k+1}})}$, proving that $\hat{\alpha}_k\geq \alpha_k$.
\end{proof}

\begin{proof}[Proof of Corollary~\ref{cor:hardtensor}]
The proof is analogous to \cite[Proof of Lemma 2.4]{Optimal}, we highlight the main
differences for completeness (since the context is somewhat different). Let
$D,\theta>0$ be arbitrary constants; we may assume that $D\geq 1$ (otherwise $U$ is
empty) and $\theta\in (0,1]$ (otherwise there are no relevant $r$). Note that since
$U$ is the set of vertices with degree $\leq D$, we have from
Corollary~\ref{lem:hardbmarginal} that the marginal distribution $\mu_{G,\lambda, U}$
is $b$-marginally bounded for $b=\min\{\frac{1}{1+\lambda},\frac{\lambda}{\lambda+(1+\lambda)^D}\}$. Let $u=|U|$,
by a standard balls-and-bins bound we have that whp $u\geq  n^{1/2}$.  By
Lemma~\ref{lem:hardspectral}, we have that whp $\mu:=\mu_{G,\lambda, U}$ is
$\eta$-spectrally independent for $\eta=\tfrac{b^2{\theta}^2}{10}\log n$.   Let
$R=\left\lceil 4\eta/b^2\right\rceil$ and note that for all sufficiently large $n$,
$\theta u \geq \theta^2 \log n\geq 4R$.

Consider an arbitrary integer $r\in [\theta u, u]$ as in the lemma statement so that $u\geq r\geq 2R$. From Theorem~\ref{thm:Optimal}, we therefore have that  $\mu$ satisfies the $r$-uniform-block  factorisation of entropy with   multiplier $C_r=\displaystyle \frac{r}{u} \frac{\sum^{u-1}_{k=0}\Gamma_k}{\sum^{u-1}_{k=u-r}\Gamma_k}\leq \frac{\sum^{u-1}_{k=0}\Gamma_k}{\sum^{u-1}_{k=u-r}\Gamma_k}$, with $\Gamma_k=\prod^{k-1}_{j=0}\alpha_j$ for $k\in [u]$ and $\alpha_k=\max\big\{0,1-\frac{4\eta}{b^2(u-k-1)}\big\}$ for $k\in [u-1]$. Then, we can apply the exact same reasoning as in \cite{Optimal} to obtain that 
\[C_r\leq \Big(\frac{u-R}{r-R}\Big)^{R}\leq \Big(\frac{2u}{r}\Big)^{R}\leq \Big(\frac{2}{\theta }\Big)^{R}.\]
Note that for $x\in (0,1)$ we have the inequality $(1/x)^{x^2}\leq \emm^x$, which for $x=\theta/2$ and using that $R\leq (\tfrac{\theta}{2})^2\log n$ gives that $C_r\leq n^\theta$, as claimed.
\end{proof}

\end{document}